\documentclass{article}[11pt]

\usepackage{amsmath}
\usepackage{mathrsfs}
\usepackage{amsthm}
\usepackage{amssymb}
\usepackage{comment}
\usepackage{cancel}
\usepackage{tikz}
\usepackage{caption}
\usepackage{graphicx}  
\usepackage{xcolor}
\usepackage{enumitem}
\usepackage{thmtools}
\usepackage{thm-restate}
\usetikzlibrary{shapes.geometric, arrows.meta}
\usetikzlibrary{positioning}

\usepackage[authoryear]{natbib}

\declaretheoremstyle[
    spaceabove=10pt,    
    spacebelow=10pt,    
    headfont=\itshape,
    bodyfont=\normalfont,
    qed=$\diamond$
]{myremark}

\declaretheoremstyle[
    spaceabove=10pt,    
    spacebelow=10pt,    
    headfont=\bfseries,
    bodyfont=\normalfont,
    qed=$\vartriangle$
]{myexample}

\declaretheorem[style=plain, name=Assumption]{assumption}

\declaretheorem[style=myexample, name=Example]{example}

\declaretheorem[style=definition, name=Definition]{definition}
\declaretheorem[style=myremark, name=Remark]{remark}

\newcommand{\ind}{\perp \!\!\!\perp}
\newcommand{\E}{\textnormal{E}}
\newcommand{\var}{\textnormal{Var}}
\newcommand{\cov}{\textnormal{Cov}}

\begin{document}     

\title{When are time series predictions causal? \\ The potential system and dynamic causal effects}
\author{Jacob Carlson and Neil Shephard\footnote{We are grateful for comments and questions from people attending seminars at Harvard, MIT, Princeton, and the University of Chicago. } \\
\textit{Department of Economics and Department of Statistics,}\\
\textit{\ Harvard University, Cambridge, MA 02138, USA} \\
}
\maketitle

\begin{abstract} The potential system is a nonparametric time series model for assessing the causal impact of moving an assignment at time $t$ on an outcome at future time $t+h$, accounting for the presence of features.  The potential system provides nonparametric content for, e.g., time series experiments, time series regression, local projection, impulse response functions and SVARs. It closes a gap between time series causality and nonparametric cross-sectional causal methods, and provides a foundation for many new methods which have causal content.       
    
\end{abstract}


\baselineskip=20pt

\section{Introduction}

Let $Y_{t+h}$ be an outcome at time $t+h$, where $h \ge 0$ is a horizon and $t$ is the current time, $X_t$ is a feature (which may be, e.g., an observed confounder), $A_t$ is an assignment and $D_{1:t-1}$ are past outcomes, features and assignments. When do time series data-based predictions, such as  
$$
\E[Y_{t+h}\mid X_t,D_{1:t-1},A_t=a_t] 
    - \E[Y_{t+h}\mid X_t,D_{1:t-1},A_t=a_t'],
$$
measure how changes in the assignment at time $t$ cause the outcomes at time $t+h$ to move?  This paper provides sufficient nonparametric conditions to answer this type of question.  

Our approach is based on defining a foundational ``potential system,'' denoted {\tt PS}. It directly connects familiar time series objects like impulse response functions to average treatment effects and, more generally, the time series causality literature to the nonparametric causal inference literature based either on population or design-based inference strategies.    

This paper is closely related to five time series papers as well as a stream of panel data papers associated with James M. Robins.  \cite%
{BojinovShephard(19)}, \cite{RambachanShephard(21)} and \cite{LinDing(25)} worked with a potential outcome based time series model, while \cite{AngristKuersteiner(11)} and \cite%
{AngristJordaKuersteiner(18)} define and work with what we call ``branch potential outcomes.''  Both potential outcomes and branch potential outcomes appear as a part of our {\tt PS} (and hence our system could be thought of as providing the primitives to these four papers).  

\cite{BojinovShephard(19)} provide many references to the literature on dynamic causal effects using potential outcome type objects.  The vast majority of the work in this type of literature on dynamic causal effects considers panel data, not
pure time series, which is the subject of this paper. The panel data literature is reviewed in, for example, \cite{HernanRobins(18)}, \cite{ArkhangelskyImbens(24)}, and \cite{ChernozhukovNeweySinghSyrgkanis(2023)}.  

A linear special case of the {\tt PS} is a structural vector autoregression (SVAR), the workhorse of modern applied linear time series methods.  Reviews of some of this work focusing on macroeconomics include \cite{KilianLutkepohl(17)}, \cite{FernandezVillaverdeRubioRamirez(10)}, \cite{StockWatson(18)}, \cite{Ramey(16)} and  \cite{JordaTaylor(24)}.  

``Granger causality'' has played an important role in time series over the last 50 years. Though inspiring, Granger causality is not really about causality, but about prediction. (``The definition of causality used above is based entirely on the predictability of some series'' \citep{Granger(69)}.)  Some of the literature on Granger causality is discussed in, for example, \cite{Kuersteiner(10)}, \cite{WhiteLu(10)} and   \cite{ShojaieFox(22)}.

\cite{HarveyDurbin(86)} tried to assess the causal impact of a one-time assignment using a time series model, applying it to assess the causal impact of the introduction of compulsion of seat belt wearing on driver deaths in the UK.  Synthetic control (\cite{AbadieGardeazabal(03)} and \cite{AbadieDiamondHainmueller(10)}) is a similar idea, but enriched. 
There a multivariate set of outcomes move together but only one is impacted by the intervention.  The multivariate data can help pin down the intervention under some assumed model.  \cite{Abadie(21)} provides a review.  

A separate ocean of work on time series causality focuses on ``control,'' where an engineer builds a system which collects data to control an output in some optimal way to their benefit.  The most famous version of this is linear/quadratic controller, e.g., \cite{Whittle(82),Whittle(83),Whittle(90a),Whittle(96)}, \cite{HansenSargent(14)} and \cite{HerbstSchorfheide(15)}.  Under the {\tt PS} the assignments can be selected to minimize expected loss given the past data --- as you would see in the control literature.  Hence the {\tt PS} bridges observational reduced form models, experiments for time series, and control models of dynamic decision making.

Much of the more modern material on control is phrased in terms of Markov decision processes (e.g., \cite{Puterman(05)}).  Sometimes the researcher uses the data to learn the Markov decision process itself.  That area is usually called reinforcement learning (e.g., \cite{SuttonBarto(18)}).   The stationary Markov version of the {\tt PS}, augmented with a loss function, again forms a bridge to this literature.    

Learning optimal dynamic treatment rules (or ``policies'' or ``regimes'') is often phrased using potential outcomes (e.g., \cite{Murphy(03)}, \cite{NieBrunskillWager(2021)}, \cite{HeckmanNavarro(2007)}, \cite{ChernozhukovNeweySinghSyrgkanis(2023)}, \cite{VivianoBradic(24)}, \cite{BradicJiZhang(2024)}). This literature connects to our work in the case of sequences of interventions, but it focuses on panel data. A notable recent exception is \cite{KitagawaWangXu(23)}, which learns optimal policies based on a single time series.   

\cite{BojinovSimchiLeviZhao(22)} and \cite{BasseDingToulis(20)} look at ``switchback designs'' to optimally learn sequences of treatment effects from time series.  There is a large other literature on sequential experiments which is not phrased in terms of potential outcomes, e.g., \cite{Efron(71)} and  \cite{GlynnJohariRasouli(20)}, as well as the substantial literature on so-called $N$-of-$1$ trials which appear prominently in, for example, the study of personalized medicine  (\cite{LilliePatayDiamantIssellTopolSchork(11)}). The design-based content of the potential system provides a nonparametric foundation for these settings as well.  Related recent work includes \cite{LiangRecht(25)}, \cite{SchnaffeAiroldi(26)} and \cite{LinDing(25)}. The latter relates potential outcomes to regression in the design-based context.   

Although phrased using potential outcomes, our system can also be written using directed acyclic graph (DAG) theory, expressed using the tools developed in the pioneering efforts of \cite{Pearl(09)} and coauthors. Important related causal graph theory topics include the ``Single World Intervention Graph'' (SWIG) associated with  \cite{RichardsonRobins(2013)}.  We use SWIG graphs to illustrate the {\tt PS} and various constraints on the equential assignment mechanism.    

The rest of this paper has six sections.  The potential system and different measures of the dynamic causal effects are defined in Section \ref{sec:Psystem}. Section \ref{Section: examples} explores several important examples of the {\tt PS} and its relationship to various other common models of causality in the time series literature.  Section \ref{Section: pred to causal} focuses on constraints on the sequential assignment mechanism and how they allow us to identify different measures of the dynamic causal effects. Section \ref{sec:ext} considers various extensions of the potential system, applying the framework to settings featuring instrumental variables, consecutive sequences of assignments, design-based causal inference, and stochastic dynamic programming (control). Section \ref{sect:conclusion} concludes. There is also an Appendix containing  proofs.  

Throughout, for any (random or deterministic) sequence $\{x_1, x_2, ..., x_T\}$ we denote for $T\ge s>t \ge 1$ the $x_{t:s} := \{x_t,...,x_{s}\}$, while $(A \ind B)\mid C$ denotes variables $A$ and $B$ are conditionally independent given $C$.

\section{Defining dynamic causality}\label{sec:Psystem}

\subsection{Defining the potential system}

The entire paper is based on the potential system, which we now define.

\begin{definition}[{\tt PS}] Start by defining two stochastic processes.  
    \begin{enumerate}
        \item The \textit{data generating process}, given by  Assumptions {\tt DGP.1} and {\tt DGP.2}.
        \item The \textit{counterfactual process}, given by  Assumptions {\tt CP.1} and {\tt CP.2}.
    \end{enumerate}
Assumption {\tt LP} links the data generating and counterfactual processes. Applying all five assumptions defines the ``potential system'' (denoted {\tt PS}). 
\end{definition}

First, the data generating process is set up using two assumptions. 

\begin{assumption}[{\tt DGP.1}]
Name the data seen at time $t$ as the split:
\[        D_{t}:=\left(X_{t}^{\mathtt{T}},A_{t}^{\mathtt{T}},Y_{t}^{\mathtt{T}}\right) ^{\mathtt{T}%
        },\quad t=1,...,T.
        \] 
We label the $X_{t} \in \mathcal{X}_t \subseteq \mathbb{R}^{d_X}$ as features; $A_t \in \mathcal{A}_t \subseteq \mathbb{R}^{d_A}$ as assignments; and $Y_{t} \in \mathcal{Y}_t \subseteq \mathbb{R}^{d_Y}$ as outcomes. Further define $\mathcal{D}_t:= \mathcal{X}_t \times \mathcal{A}_t \times \mathcal{Y}_t$,  $\mathcal{D}_{t:s} := \prod_{j=t}^{s} \mathcal{D}_j,$ and $\mathcal{A}_{t:s} := \prod_{j=t}^{s} \mathcal{A}_j,$ for $s \ge t$. 
\end{assumption}

\begin{remark}[Feature interpretation]
Depending on the assumptions made about them, features $X_{t}$ can play the role of observed confounders (explored throughout most of the paper), instruments (explored in Section \ref{sec:iv}), or whatever else may be suitable to a given empirical setting.
\end{remark}

\begin{remark}[Feature indexing]
The contemporaneous time indexing of features in Assumption {\tt DGP.1} is a convention. For example, features could also be characterized by the random variable $X_t^*$ where $X_{t}  = X_{t-1}^* $.
\end{remark}

\begin{assumption}[{\tt DGP.2}] Assume the time $t$ assignment is generated by the ``sequential assignment mechanism'' ({\tt SAM}),
$$
A_t = \alpha_t\left(D_{1:t-1}, X_t,V_t\right),\quad t=1,...,T,
$$
where $V_t$ is crystallized by time $t$, the $V_t \mid X_t, D_{1:t-1}$ is stochastic and $V_t \ind D_{1:t-1}$. Throughout, assume the function
$$\alpha_t := \{\alpha_t(d_{1:t-1},x_t,v_t):d_{1:t-1}\in \mathcal{D}_{1:t-1},x_t\in \mathcal{X}_t,v_t \in \mathcal{V}_t\},
$$ 
is deterministic with respect to knowledge at time 0. 

\end{assumption}

Second, the counterfactual process is set up using two assumptions. 

\begin{assumption}[{\tt CP.1}] 
The time $t$ ``potential feature'' and ``potential outcome'' are collected as
\[
Z_{t}(a_{1:T}):=\left\{ X_{t}(a_{1:T}),Y_{t}(a_{1:T})\right\}, \quad a_{1:T} \in \mathcal{A}_{1:T},\quad X_{t}(a_{1:T}) \in \mathcal{X}_t,\quad Y_{t}(a_{1:T}) \in \mathcal{Y}_t,\quad t=1,...,T,
\]
where $a_{1:T}$ is a possible assignment path that obeys both of:   
    \begin{enumerate}[labelsep=1em, leftmargin=5em]
    \item[{\tt CP.1a}] (Non-anticipation) For all $a_{1:T}$ and $a_{1:T}'\in \mathcal{A}_{1:T}$, the 
    $$
    Z_{t}(a_{1:T})=Z_{t}(a_{1:t},a_{t+1:T}^{\prime }).$$ We write this in shorthand as  
    $Z_{t}(a_{1:t}).$
    \item[{\tt CP.1b}] (Triangularity) For all $a_{1:t},a_{1:t}'$, the 
    $$
    \begin{array}{lll}
    X_t(a_{1:t}) &= X_t(a_{1:t-1},a_t').
    \end{array} 
    $$
    We write this in shorthand as $X_{t}(a_{1:t-1})$.
\end{enumerate} 
Per {\tt CP.1a} and {\tt CP.1b}, we simplify the notation to 
\[
Z_{t}(a_{1:t}) = \left\{ X_{t}(a_{1:t-1}), Y_{t}(a_{1:t}) \right\},
\]
and collect the path of counterfactuals for all $T$ periods as  
\[
Z_{1:T}(a_{1:T}) := \left\{ Z_{1}(a_{1}), \ldots, Z_{T}(a_{1:T}) \right\}.
\]
\end{assumption} 

\begin{remark}[Non-interference]
Assumption {\tt CP.1a} implies $Z_t(a_{1:T})$ is realized at time $t$ and cannot depend on future assignments $a_{t+1:T}$. This assumption rules out a time series form of what \cite{Cox(58book)} generically called ``interference.'' The use of non-anticipation arguments as important criteria for temporal causation appears in, for example, \cite{Granger(80cause)} and \cite{RambachanShephard(21)}.
\end{remark}

The left-hand side of Figure \ref{fig:causFig} visualizes the counterfactual paths of potential outcomes and confounders defined in {\tt CP.1} for binary assignments. The right-hand side shows $Z_{1:3}$ corresponding to $A_{1:3}=(1,1,0)$, highlighting an assigned path in bold.

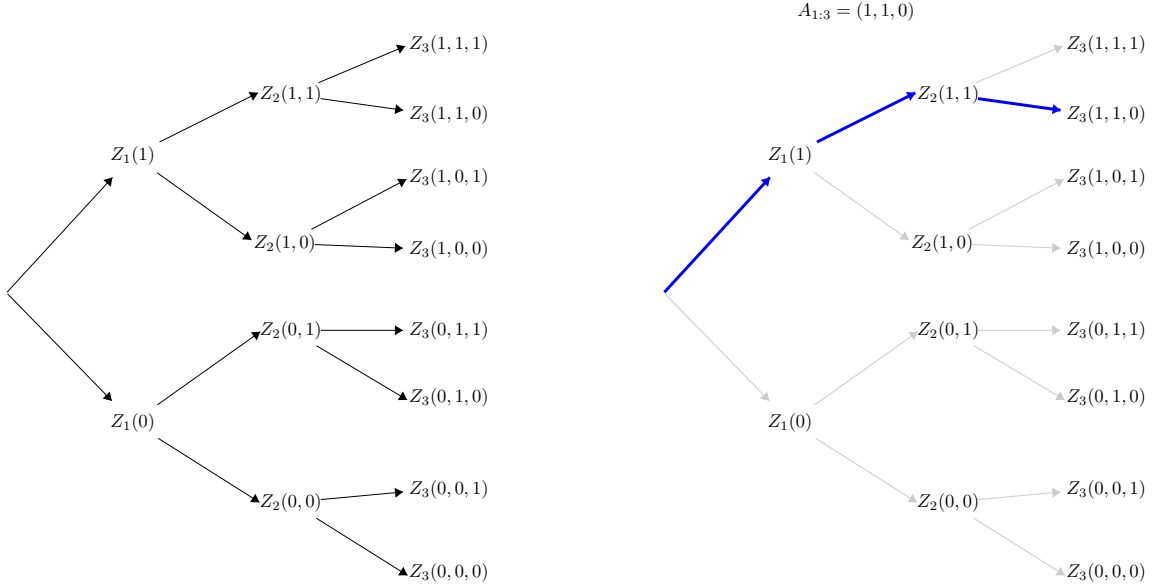
\begin{figure}[t]
\begin{center}
\scalebox{0.7}{\begin{tikzpicture}[scale=0.19]
    \tikzstyle{every node}+=[inner sep=0pt]
      \draw (23.4,-14.4) node {$Z_1(1)$};
      \draw (23.4,-41.1) node {$Z_1(0)$};
      \draw (39.2,-8.3) node {$Z_2(1,1)$};
      \draw (38.6,-23.2) node {$Z_2(1,0)$};
      \draw (39.2,-31.9) node {$Z_2(0,1)$};
      \draw (39.2,-49.1) node {$Z_2(0,0)$};
      \draw (55,-3.3) node {$Z_3(1,1,1)$};
      \draw (55,-10.3) node {$Z_3(1,1,0)$};
      \draw (55,-16.6) node {$Z_3(1,0,1)$};
      \draw (55,-23.8) node {$Z_3(1,0,0)$};
      \draw (55,-31.9) node {$Z_3(0,1,1)$};
      \draw (55,-38.6) node {$Z_3(0,1,0)$};
      \draw (55,-56.1) node {$Z_3(0,0,0)$};
      \draw (55,-47.8) node {$Z_3(0,0,1)$};
      \draw [black] (10.83,-28.09) -- (21.37,-16.61);
      \fill [black] (21.37,-16.61) -- (20.46,-16.86) -- (21.2,-17.54);
      \draw [black] (41.77,-33.44) -- (50.83,-38.86);
      \fill [black] (50.83,-38.86) -- (50.4,-38.02) -- (49.88,-38.88);
      \draw [black] (42.19,-48.83) -- (50.41,-48.07);
      \fill [black] (50.41,-48.07) -- (49.57,-47.65) -- (49.66,-48.64);
      \draw [black] (41.74,-50.69) -- (50.86,-56.41);
      \fill [black] (50.86,-56.41) -- (50.45,-55.56) -- (49.91,-56.41);
      \draw [black] (41.96,-7.13) -- (50.64,-3.47);
      \fill [black] (50.64,-3.47) -- (49.71,-3.32) -- (50.09,-4.24);
      \draw [black] (42.17,-8.72) -- (50.43,-9.88);
      \fill [black] (50.43,-9.88) -- (49.71,-9.27) -- (49.57,-10.27);
      \draw [black] (41.25,-21.8) -- (50.75,-16.8);
      \fill [black] (50.75,-16.8) -- (49.81,-16.73) -- (50.27,-17.61);
      \draw [black] (41.6,-23.32) -- (50.4,-23.68);
      \fill [black] (50.4,-23.68) -- (49.62,-23.15) -- (49.58,-24.15);
      \draw [black] (26.08,-13.06) -- (35.92,-8.14);
      \fill [black] (35.92,-8.14) -- (34.98,-8.05) -- (35.42,-8.95);
      \draw [black] (25.84,-16.14) -- (35.26,-22.86);
      \fill [black] (35.26,-22.86) -- (34.9,-21.99) -- (34.32,-22.8);
      \draw [black] (25.85,-39.36) -- (36.15,-32.04);
      \fill [black] (36.15,-32.04) -- (35.21,-32.09) -- (35.79,-32.91);
      \draw [black] (25.94,-42.7) -- (36.06,-49.1);
      \fill [black] (36.06,-49.1) -- (35.65,-48.25) -- (35.12,-49.09);
      \draw [black] (10.89,-28.25) -- (21.31,-38.95);
      \fill [black] (21.31,-38.95) -- (21.11,-38.03) -- (20.39,-38.73);
      \draw [black] (42.2,-31.9) -- (50.4,-31.9);
      \fill [black] (50.4,-31.9) -- (49.6,-31.4) -- (49.6,-32.4);
    \end{tikzpicture} } \hspace{2cm} 
\scalebox{0.7}{\begin{tikzpicture}[scale=0.19]
        \tikzstyle{every node}+=[inner sep=0pt]
        \draw (30, 0) node {$A_{1:3} = (1,1,0)$};
        \draw (23.4,-14.4) node {$Z_1(1)$};
        \draw (23.4,-41.1) node {$Z_1(0)$};
        \draw (39.2,-8.3) node {$Z_2(1,1)$};
        \draw (38.6,-23.2) node {$Z_2(1,0)$};
        \draw (39.2,-31.9) node {$Z_2(0,1)$};
        \draw (39.2,-49.1) node {$Z_2(0,0)$};
        \draw (55,-3.3) node {$Z_3(1,1,1)$};
        \draw (55,-10.3) node {$Z_3(1,1,0)$};
        \draw (55,-16.6) node {$Z_3(1,0,1)$};
        \draw (55,-23.8) node {$Z_3(1,0,0)$};
        \draw (55,-31.9) node {$Z_3(0,1,1)$};
        \draw (55,-38.6) node {$Z_3(0,1,0)$};
        \draw (55,-56.1) node {$Z_3(0,0,0)$};
        \draw (55,-47.8) node {$Z_3(0,0,1)$};
        \draw [ultra thick,blue] (10.83,-28.09) -- (21.37,-16.61);
        \fill [ultra thick,blue] (21.37,-16.61) -- (20.46,-16.86) -- (21.2,-17.54);
        \draw [gray!40] (41.77,-33.44) -- (50.83,-38.86);
        \fill [gray!40] (50.83,-38.86) -- (50.4,-38.02) -- (49.88,-38.88);
        \draw [gray!40] (42.19,-48.83) -- (50.41,-48.07);
        \fill [gray!40] (50.41,-48.07) -- (49.57,-47.65) -- (49.66,-48.64);
        \draw [gray!40] (41.74,-50.69) -- (50.86,-56.41);
        \fill [gray!40] (50.86,-56.41) -- (50.45,-55.56) -- (49.91,-56.41);
        \draw [gray!40] (41.96,-7.13) -- (50.64,-3.47);
        \fill [gray!40] (50.64,-3.47) -- (49.71,-3.32) -- (50.09,-4.24);
        \draw [ultra thick,blue] (42.17,-8.72) -- (50.43,-9.88);
        \fill [ultra thick,blue] (50.43,-9.88) -- (49.71,-9.27) -- (49.57,-10.27);
        \draw [gray!40] (41.25,-21.8) -- (50.75,-16.8);
        \fill [gray!40] (50.75,-16.8) -- (49.81,-16.73) -- (50.27,-17.61);
        \draw [gray!40] (41.6,-23.32) -- (50.4,-23.68);
        \fill [gray!40] (50.4,-23.68) -- (49.62,-23.15) -- (49.58,-24.15);
        \draw [ultra thick,blue] (26.08,-13.06) -- (35.92,-8.14);
        \fill [ultra thick,blue] (35.92,-8.14) -- (34.98,-8.05) -- (35.42,-8.95);
        \draw [gray!40] (25.84,-16.14) -- (35.26,-22.86);
        \fill [gray!40] (35.26,-22.86) -- (34.9,-21.99) -- (34.32,-22.8);
        \draw [gray!40] (25.85,-39.36) -- (36.15,-32.04);
        \fill [gray!40] (36.15,-32.04) -- (35.21,-32.09) -- (35.79,-32.91);
        \draw [gray!40] (25.94,-42.7) -- (36.06,-49.1);
        \fill [gray!40] (36.06,-49.1) -- (35.65,-48.25) -- (35.12,-49.09);
        \draw [gray!40] (10.89,-28.25) -- (21.31,-38.95);
        \fill [gray!40] (21.31,-38.95) -- (21.11,-38.03) -- (20.39,-38.73);
        \draw [gray!40] (42.2,-31.9) -- (50.4,-31.9);
        \fill [gray!40] (50.4,-31.9) -- (49.6,-31.4) -- (49.6,-32.4);
      \end{tikzpicture}}
\end{center}
\caption{The left figure shows all the potential outcome paths for $T=3$.
The right figure shows the observed outcome path $Z_{1:3}(A_{1:3})$ where $%
A_{1:3}=(1,1,0)^{\mathtt{T}}$, indicated by the thick blue line. The gray
arrows indicate the missing data. }
\label{fig:causFig}
\end{figure}

\begin{assumption}[{\tt CP.2}]\label{ass:CP2} Write the ``potential branch'' at time $t+h$ as 
$$D_{t,h}(a_t):=\{X_{t,h}(a_t),A_{t,h}(a_t),Y_{t,h}(a_t)\},\quad h=0,1,...,H,
$$ 
a system counterfactual.  It  corresponds to the assignment at time $t$ being set to $a_t$ and recording the system at horizon $h$ periods later. 
Assume the ``branch assignment'' at horizon $h$ is 
$$
A_{t, h}\left(a_t\right):= \begin{cases}
\begin{array}{lll}
     a_t , & h=0 \\
     \alpha_{t+h}\left(D_{1:t-1},D_{t,0:h-1}\left(a_t\right), X_{t,h}(a_t),V_{t+h}\right), & h > 0.
\end{array}
\end{cases} 
$$
Assume the ``branch potential outcome'' and ``branch potential feature'' at horizon $h$ are
$$
Z_{t,h}(a_t):= \{X_{t,h}(a_t)^\mathtt{T}, Y_{t,h}(a_t)^\mathtt{T}\}^\mathtt{T}:= \begin{cases}
\begin{array}{lll}
     \{X_{t}^\mathtt{T},Y_{t}(A_{1:t-1}, a_{t})^\mathtt{T}\}^\mathtt{T} , & h=0 \\
     \{X_{t+h}(A_{1:t-1}, A_{t,0:h-1}(a_{t}))^\mathtt{T}, Y_{t+h}(A_{1:t-1}, A_{t,0:h}(a_{t}))^\mathtt{T}\}^\mathtt{T}, & h > 0.
\end{array}
\end{cases} 
$$
    
\end{assumption}

\begin{remark}[Branch potential outcomes]
Definition {\tt CP.2} defines the branch potential outcomes $Y_{t,h}(a_{t})$.    \cite{AngristKuersteiner(11)} and \cite{AngristJordaKuersteiner(18)} worked directly with branch potential outcomes, without spelling out an underlying {\tt PS}.   \cite{BojinovShephard(19)} worked with potential outcomes $Y_{t+h}(a_{t:t+h})$ assuming the assignments were sequentially randomized and there were no confounders.
\end{remark}

The left-hand side of Figure \ref{fig:counter} visualizes the counterfactual paths of potential branches, again for binary assignments. The right-hand side shows $D_{t:t+2}$ corresponding to $A_t=1$, highlighting the assigned path.  
 
\begin{figure}[t]
\begin{center}
\scalebox{0.7}{\begin{tikzpicture}[scale=0.19]
    \tikzstyle{every node}+=[inner sep=0pt]
      \draw (23.4,-14.4) node {$D_{t,0}(1)$};
      \draw (23.4,-41.1) node {$D_{t,0}(0)$};
      \draw (39.2,-8.3) node {$D_{t,1}(1)$};
      \draw (39.2,-49.1) node {$D_{t,1}(0)$};
      \draw (55,-3.3) node {$D_{t,2}(1)$};
      \draw (55,-56.1) node {$D_{t,2}(0)$};

      \draw [black] (10.83,-28.09) -- (21.37,-16.61);
        \fill [black] (21.37,-16.61) -- (20.46,-16.86) -- (21.2,-17.54);

        \draw [black] (26.08,-13.06) -- (35.92,-8.14);
        \fill [black] (35.92,-8.14) -- (34.98,-8.05) -- (35.42,-8.95);

        \draw [black] (41.77,-6.44) -- (50.64,-3.47);
        \fill [black] (50.64,-3.47) -- (49.71,-3.32) -- (50.09,-4.24);

        \draw [black] (41.74,-50.69) -- (50.86,-56.41);
        \fill [black] (50.86,-56.41) -- (50.45,-55.56) -- (49.91,-56.41);

        \draw [black] (25.94,-42.7) -- (36.06,-49.1);
        \fill [black] (36.06,-49.1) -- (35.65,-48.25) -- (35.12,-49.09);

        \draw [black] (10.89,-28.25) -- (21.31,-38.95);
        \fill [black] (21.31,-38.95) -- (21.11,-38.03) -- (20.39,-38.73);

    \end{tikzpicture} } \hspace{2cm} 
\scalebox{0.7}{\begin{tikzpicture}[scale=0.19]
        \tikzstyle{every node}+=[inner sep=0pt]
        \draw (30, 0) node {$A_{t} = 1$};
        \draw (23.4,-14.4) node {$D_{t,0}(1)$};
        \draw (23.4,-41.1) node {$D_{t,0}(0)$};
        \draw (39.2,-8.3) node {$D_{t,1}(1)$};
        \draw (39.2,-49.1) node {$D_{t,1}(0)$};
        \draw (55,-3.3) node {$D_{t,2}(1)$};
        \draw (55,-56.1) node {$D_{t,2}(0)$};

        \draw [ultra thick,blue] (10.83,-28.09) -- (21.37,-16.61);
        \fill [ultra thick,blue] (21.37,-16.61) -- (20.46,-16.86) -- (21.2,-17.54);

        \draw [ultra thick,blue] (26.08,-13.06) -- (35.92,-8.14);
        \fill [ultra thick,blue] (35.92,-8.14) -- (34.98,-8.05) -- (35.42,-8.95);

        \draw [ultra thick,blue] (41.77,-6.44) -- (50.64,-3.47);
        \fill [ultra thick,blue] (50.64,-3.47) -- (49.71,-3.32) -- (50.09,-4.24);

        \draw [gray!40] (41.74,-50.69) -- (50.86,-56.41);
        \fill [gray!40] (50.86,-56.41) -- (50.45,-55.56) -- (49.91,-56.41);

        \draw [gray!40] (25.94,-42.7) -- (36.06,-49.1);
        \fill [gray!40] (36.06,-49.1) -- (35.65,-48.25) -- (35.12,-49.09);

        \draw [gray!40] (10.89,-28.25) -- (21.31,-38.95);
        \fill [gray!40] (21.31,-38.95) -- (21.11,-38.03) -- (20.39,-38.73);
         \end{tikzpicture}}
\end{center}
\caption{Time-$t$ system counterfactual paths. \ The left figure shows all
the potential branches $D_{t,h}(a_{t})$ paths for horizon $h=0,1,2$ and $%
a_{t}\in \{0,1\}$. The right figure shows the observed outcome path $%
D_{t:t+2}=D_{t,0:2}(A_{t})$ where $A_{t}=1$, indicated by the thick blue
line. The gray arrows indicate the missing data. }
\label{fig:counter}
\end{figure}
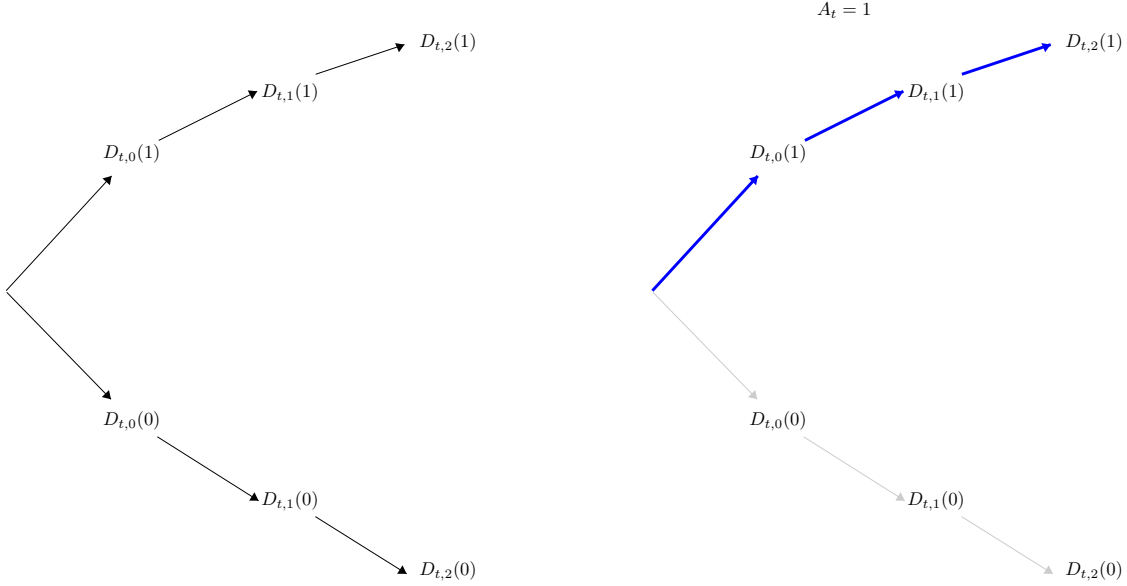

Finally, the data generating and counterfactual processes are linked by one assumption, completing the definition of the potential system.  

\begin{assumption}[{\tt LP}] Assume $$
Z_{1:T} = Z_{1:T}(A_{1:T}),$$ 
that is the data generating process and counterfactual process are ``consistent.'' 
    
\end{assumption}

\begin{remark}[``Consistency'']
Assumptions {\tt LP} and {\tt CP.1} enforce ``system consistency'': $Z_t = Z_{t,0}(A_{t}) = Z_t(A_{1:t})$. Assumption {\tt LP} and Assumption {\tt CP.2}  imply $Z_{t+h} = Z_{t,h}(A_{t})$ for $h>0$. This is the system version of the ``no hidden treatments'' component of the Stable Unit Treatment Value Assumption (SUTVA) formalized by \cite{Rubin(80)}; see also \cite{Robins(86)} and \cite{ImbensRubin(15)}.  
\end{remark}

\begin{remark}[Multi-period assignments] Recall {\tt CP.2} moves the single period assignment $a_t$.  Section \ref{sec:mpa} broadens the {\tt PS} definition, replacing {\tt CP.2} with {\tt CP.2$^\prime$}, which moves $a_{t:t+s}$, multi-period assignments, where $s\ge0$.  The rest of the system is unchanged and no new fundamental ideas are needed.     
\end{remark}

\subsection{Defining causality with respect to a {\tt PS}}

Having set up the {\tt PS}, it is now possible to define what a dynamic causal effect is and, in turn, how it can be summarized in the case it is stochastic. 

\begin{definition}[Dynamic causal effects] Assume a {\tt PS}. The ``dynamic causal effect'' is $$Y_{t,h}(a_{t})-Y_{t,h}(a_{t}^{\prime }),$$ the effect of moving the time $t$ assignment from $a_{t}^{\prime } \in \mathcal{A}_{t}$ to $a_{t}\in \mathcal{A}_{t}$ on the time $t+h$ outcome, where the horizon is $h \ge 0$. 
\end{definition}

We may summarize the dynamic causal effect in a number of ways. 

\begin{definition}[Describing dynamic causal effects]\label{defn:causaleffect} Assume the {\tt PS} is in $L^1$.  Then: 

\begin{enumerate} 

\item The \textquotedblleft average treatment effect\textquotedblright\ is
    $$
    {\tt ATE}_{t,h}(a_t,a_t') :=\E[\{Y_{t,h}(a_{t})-Y_{t,h}(a_{t}^{\prime })\}].
    $$

\item The \textquotedblleft conditional average treatment effect\textquotedblright\ is
    \[
    {\tt CATE}_{t,h}(a_t,a_t'):=\E[\{Y_{t,h}(a_{t})-Y_{t,h}(a_{t}^{\prime })\}\mid X_t].
    \]%
    
\item The \textquotedblleft filtered treatment 
    effect\textquotedblright\ is
    $$
    {\tt FTE}_{t,h}(a_t,a_t'):=\E[\{Y_{t,h}(a_{t})-Y_{t,h}(a_{t}^{\prime })\}\mid D_{1:t-1}].
    $$
    
\item The \textquotedblleft conditional filtered treatment
    effect\textquotedblright\ is
    $$
    {\tt CFTE}_{t,h}(a_t,a_t'):=\E[\{Y_{t,h}(a_{t})-Y_{t,h}(a_{t}^{\prime })\}\mid X_t,D_{1:t-1}].
    $$
\end{enumerate}
\end{definition}

\begin{remark}[Total versus direct dynamic causal effects]
The dynamic causal effect at horizon one (for example) is the ``total'' dynamic causal effect of moving $a_t'$ to $a_t$:
\begin{align*}
    Y_{t,1}(a_t) - Y_{t,1}(a_t') =& Y_{t+1}(A_{1:t-1},a_t,A_{t,1}(a_t)) - Y_{t+1}(A_{1:t-1},a_t',A_{t,1}(a_t')),
\end{align*}
which captures how moving $a_t'$ to $a_t$ also affects future assignments.  This compares to the ``direct'' dynamic causal effect one could be interested in,
\begin{align*}
     Y_{t+1}(A_{1:t-1},a_t,A_{t+1}) - Y_{t+1}(A_{1:t-1},a_t',A_{t+1}),
\end{align*}
which ignores how moving $a_t'$ to $a_t$ affects future assignments, and is expressed directly in terms of the potential outcomes (not branch potential outcomes).  When $h=0$ the total and direct dynamic causal effects are the same.
\end{remark}

\begin{remark}[Marginal dynamic causal effects]
The marginal dynamic causal effect, if it exists, is defined as 
$$
\frac{\partial Y_{t,h}(a_t)}{\partial a_t}.
$$
The average, conditional, filtered, and conditional filtered treatment effects have obvious marginal versions, taking expectations of the marginal causal effect.  The simplest case is $\partial Y_{t,0}(a_t)/\partial a_t = \partial Y_{t}(A_{1:t-1},a_t)/\partial a_t$, expressing the marginal causal effect in terms of derivatives of the potential outcome. Results for $h>0$ can be calculated recursively.
\end{remark}

\begin{remark}[Causal measures in context] 
Definition ${\tt ATE}_{t,h}(a_t,a_t')$ is typically called the ``impulse response function'' (e.g., \cite{Sims(80)}) in time series. Average treatment effects appear frequently in, for example, randomized control trials, e.g., \cite{ImbensRubin(15)}.  Definition ${\tt CATE}_{t,h}(a_t,a_t')$ appears frequently in cross-sectional observational causal studies, e.g., see \cite{ImbensRubin(15)}.  Definition ${\tt FTE}_{t,h}(a_t,a_t')$ is typically called the ``generalized impulse response function'' (e.g., \cite{KoopPesaranPotter(96)}) in time series.
\end{remark}

\section{Examples of the potential system}\label{Section: examples}

The following are important special cases of the {\tt PS}. Going forward, for notational convenience, we define $D_t(a_{1:t}) := \{X_t(a_{1:t-1}),a_t,Y_t(a_{1:t})\}$ and $D_{1:t}(a_{1:t}) := \{D_1(a_1),...,D_t(a_{1:t})\}$.

\subsection{Structural equation model potential systems}

We begin by discussing the nonparametric structural equation model (SEM) potential system, a highly general example of a {\tt PS} that adds useful additional structure to the potential outcomes and features.

\begin{example}[Triangular nonparametric SEM {\tt PS}]\label{ex:npsem}
Assume a {\tt PS}. The sequential triangular nonparametric simultaneous system sets the time $t$ potential feature and potential outcome as 
\begin{eqnarray*}
X_{t}(a_{1:t-1}) &=&\chi_t(D_{1:t-1}(a_{1:t-1}),U_{t}), \\
Y_{t}(a_{1:t}) &=&\gamma_t(D_{1:t-1}(a_{1:t-1}),X_{t}(a_{1:t-1}),a_t,W_{t}),
\end{eqnarray*}
where all  $\varepsilon_t := (U_t^\mathtt{T},V_t^\mathtt{T},W_t^\mathtt{T})^\mathtt{T}$ are crystallized by time $t$, are independent over time, the $\varepsilon_t \mid D_{1:t-1}$ are stochastic, the $\varepsilon_t \ind D_{1:t-1}$ and the functions $$\chi_t := \{\chi_t(d_{1:t-1},u_t):d_{1:t-1}\in \mathcal{D}_{1:t-1},u_t \in \mathcal{U}\}
$$
and 
$$
\gamma_t := \{\gamma_t(d_{1:t-1},x_t,a_t,w_t):d_{1:t-1}\in \mathcal{D}_{1:t-1},x_t \in \mathcal{X}_t,a_t\in \mathcal{A}_t,w_t \in \mathcal{W}_t\}
$$
are deterministic with respect to knowledge at time 0 for each $t=1,...,T$.  
\end{example}

This model can be viewed as requiring that Assumption {\tt DGP.2} defines a nonparametric structural equation model (NPSEM)  or structural causal model (SCM), formalized by \cite{Pearl(95),Pearl(09)}.  This framework for causal inference has many direct ties to potential outcomes frameworks for inference on counterfactuals (see, e.g.,  \cite{Imbens(2020)} or \cite{Geffner(2022)} for linkages).

\begin{remark}[Lucas critique]
Notice as $a_{t}$ moves, $(U_{t:T},V_{t:T},W_{t:T})$ (the primitives which drive the system) do not, and the functional forms $\chi_t,\gamma_t$ in the counterfactual process do not change (nor does $\alpha_t$). Taken together, this is a system version of assuming a sufficiently rich structure to avoid the Lucas critique in economics (see, e.g., \cite{Lucas(1976)}, \cite{McKayWolf(23)}, \cite{Sargent(2025)}).
\end{remark}

\subsection{Linear potential systems}

We may further specialize Example \ref{ex:npsem} by incorporating linearity, delivering the homogeneous linear Markov {\tt PS}. 

\begin{example}[Homogeneous linear Markov {\tt PS}]\label{ex:hlmps}
Assume a {\tt PS} for which the sequential assignment mechanism is given by
\[
A_t = \alpha_1 D_{t-1} + \alpha_0 X_t + \Gamma V_t
\]
for conformable $\alpha_0,\alpha_1,\Gamma$ and the potential outcome and feature are given by
\[
\begin{array}{lll}
X_{t}(a_{1:t-1}) & = & \chi_1 D_{t-1}(a_{1:t-1}) + \Delta U_{t} \\ 
Y_{t}(a_{1:t}) & = & \gamma_1 D_{t-1}(a_{1:t-1})+\gamma_{0,X} X_{t}(a_{1:t-1}) + \gamma_{0,A} a_t + \Omega W_{t},%
\end{array}%
\]%
for conformable $\chi_1,\Delta,\gamma_{0,X},\gamma_{0,A},\gamma_1,\Omega$, and for which $\{\varepsilon_{t}\}_{t\geq 1}:=\left\{(U_t^{\mathtt{T}}, V_t^{\mathtt{T}},  W_t^{\mathtt{T}})\right\}^{\mathtt{T}}_{t\geq 1} \overset{ind}{\sim}$.
\end{example}
Under the model of Example \ref{ex:hlmps}, Assumption {\tt LP} implies that the DGP is 
\[
\begin{array}{lll}
X_{t} & = & \chi_1 D_{t-1}+\Delta U_{t}, \\
A_{t} & = & \alpha_0 X_t + \alpha_1 D_{t-1}+\Gamma V_{t}, \\
Y_{t} & = & \gamma_{0,X} X_{t} + \gamma_{0,A} A_t 
+ 
\gamma_1 D_{t-1} + \Omega W_t.
\end{array}%
\]%
We may thus compactly write the DGP as a VAR(1):
$$
D_{t}=\phi D_{t-1}+B\varepsilon _{t}
$$ 
where
\[
\phi := \left(\begin{array}{c}\chi_1 \\ \alpha_1 + \alpha_0 \chi_1 \\ \gamma_1 + \gamma_{0,X} \chi_1 + \gamma_{0,A}(\alpha_1 + \alpha_0 \chi_1 )\end{array} \right),\quad B:=\left( 
\begin{array}{ccc}
\Delta & 0 & 0\\ 
\alpha_0\Delta  & \Gamma & 0\\ 
(\gamma_{0,X} + \gamma_{0,A} \alpha_0)\Delta & \gamma_{0,A}\Gamma  & \Omega
\end{array}%
\right).
\]
Writing the system this way makes clear it is Markovian.  

For the counterfactual process, we start by writing 
$$
\begin{array}{lll}
D_{t,0}(a_{t})&:=&\left( 
\begin{array}{c}
X_{t,0}(a_{t}) \\
A_{t,0}(a_{t}) \\
Y_{t,0}(a_{t})%
\end{array}%
\right)
\end{array}.
$$
As $A_{t,0}(a_{t})=a_{t}$, 
\[
D_{t,0}(a_{t})= \left( 
\begin{array}{c}
\chi_1 \\
0 \\
\gamma_1 + \gamma_{0,X} \chi_1  
\end{array}%
\right) D_{t-1}+\left( 
\begin{array}{c}
0 \\ 
I \\ 
\gamma_{0,A} 
\end{array}%
\right) a_{t} +
\left( 
\begin{array}{ccc}
\Delta  & 0 & 0\\ 
0 & 0 & 0\\ 
\gamma_{0,X} \Delta & 0& \Omega%
\end{array}%
\right) \varepsilon _{t},
\]%
and, for $h=1,2,...,H$, we define
\[
D_{t,h}(a_{t}):=\phi D_{t,h-1}(a_{t})+B\varepsilon _{t+h}.
\]%

The dynamic causal effect on all variables at horizion $h\geq 0$ is given by
    \begin{eqnarray*}
    D_{t,h}(a_{t})-D_{t,h}(a_{t}^{\prime }) &=&\phi \left\{
    D_{t,h-1}(a_{t})-D_{t,h-1}(a_{t}^{\prime })\right\}  =\phi ^{h}\left\{ D_{t,0}(a_{t})-D_{t,0}(a_{t}^{\prime })\right\}  \\
    &=&
    \Psi_h( a_{t}-a_{t}^{\prime }),\quad \Psi_h:= \phi ^{h}\left( 
    \begin{array}{c}
    0 \\ 
    I \\ 
    \gamma_{0,A}  
    \end{array}%
    \right), 
    \end{eqnarray*}
    which is non-stochastic.  Thus (abusing notation slightly)
    $$
    D_{t,h}(a_{t})-D_{t,h}(a_{t}^{\prime }) = {\tt ATE}_{t,h}(a_t,a_t') = {\tt CATE}_{t,h}(a_t,a_t') = {\tt FTE}_{t,h}(a_t,a_t') = {\tt CFTE}_{t,h}(a_t,a_t').
    $$

\begin{remark}[Slutzky-Frisch paradigm]\label{remark:sf}
Assume the homogeneous linear Markov {\tt PS} from Example \ref{ex:hlmps}.
\begin{enumerate}
    \item If $\{\varepsilon _{t}\}_{t \ge 1} \overset{iid}{\sim} $, then $\{D_t\}_{t \ge 1}$ can be written as a VAR(1). It can also be written as a ``structural'' VAR(1) --- or SVAR(1) \citep{Sims(80)} --- which is
\[
\tilde BD_{t}=\tilde \phi D_{t-1}+\varepsilon _{t}
\]
for $\tilde B := B^{-1}$ and $\tilde\phi := B^{-1}\phi$. Through recursive substitution of
the VAR(1) process, we can also write the system in a ``structural vector moving average'' (SVMA) representation:  
$$
D_{t+h} = \phi^{t+h-1}D_1+ \sum_{j=0}^{t+h-2} \Theta_j \varepsilon_{t+h-j},\quad \Theta_h := \phi^h B.
$$
If the absolute value of the largest eigenvalue of $\phi$ is strictly less than one, the $\{\varepsilon_s : s \in \mathbb{Z}\}$ is in $L^2$, and the process holds infinitely in the past, then 
$$D_{t+h} = \sum_{j=0}^{\infty} \Theta_j \varepsilon_{t+h-j}$$ 
exists; this representation is often labeled a SVMA($\infty$) process.  It appears at the heart of the so-called ``Slutzky-Frisch impulse-propagation paradigm'' in macroeconomics (see, e.g., \cite{StockWatson(18)} for an econometric overview).
    \item For simplicity of exposition, consider a scalar outcome, assignment and feature, and let $e_j$ be the $j$-th column of a $3 \times 3$ identity matrix. In the Slutzky-Frisch paradigm, the scalar function $h \mapsto e_3^{\mathtt{T}}\Theta_h e_2$ is called the impulse response function (IRF) for $V_t$ (the ``shock'' to assignments, not the assignment itself) on the outcome at horizon $h$, as
    \[ e_3^{\mathtt{T}}\Theta_h e_2=\E[Y_{t+h}\mid V_t = 1]-\E[Y_{t+h}\mid V_t = 0] = e_3^{\mathtt{T}}\Psi_h \Gamma. \]
    Notice further that the causal effect of moving $A_t$ from 0 to 1 on the outcome for any given $h \geq 0$, the scalar function $h\mapsto e_3^{\mathtt{T}}\Psi_h$ (recalling that $\Psi_h\in\mathbb{R}^3$), can then be considered the ``relative IRF'': the $e_2^\mathtt{T}\Theta_0 e_2 = \Gamma = \E[A_{t}\mid V_t=1]-\E[A_t \mid V_t = 0]$, and so almost surely
    \[Y_{t,h}(1)-Y_{t,h}(0)=e_3^{\mathtt{T}}\Psi_h = \frac{\E[Y_{t+h}\mid V_t = 1]-\E[Y_{t+h}\mid V_t = 0]}{\E[A_{t}\mid V_t=1]-\E[A_t \mid V_t = 0]} .\]
    This observation naturally extends to vector-valued outcomes, assignments, and features. 
\end{enumerate}
\end{remark}

\subsection{Potential systems without features}

Another important special case of the {\tt PS} occurs when there are no features and the assignments are independent through time. 
We may further specialize Example \ref{ex:npsem} to explore this setting: the Homogeneous Markov news impact {\tt PS}.  

\begin{example}[Homogeneous Markov news impact {\tt PS}]\label{ex:news}
Assume a {\tt PS}. The homogeneous Markov news impact {\tt PS} has 
\[
    \begin{array}{ll}
    X_{t}(a_{1:t-1}) & = 0 \\ 
    A_t & = V_t \\
    Y_{t}(a_{1:t}) & = \gamma(Y_{t-1}(a_{1:t-1}),a_{t},W_t)
    \end{array}%
    \]%
    where $\{(V_t,W_t)\}_{t\ge 1}$ is an independent sequence and $\gamma$ is a non-random function known at time 0. 
\end{example}

The DGP under Example \ref{ex:news} is then
\[
    \begin{array}{ll}
    A_t & = V_t \\
    Y_t &= \gamma(Y_{t-1},A_{t},W_t).
    \end{array}%
    \]%
Note further that the $h=0$ potential branch and $h > 0$ potential branch are, respectively,
$$
Y_{t,0}(a_t) = \gamma(Y_{t-1},a_{t},W_t),
\quad Y_{t,h}(a_t) = \gamma(Y_{t,h-1}(a_t),V_{t+h},W_{t+h}),
$$
observing that $A_{t+h}(a_t) = V_{t+h}$ for all $h > 0$. The causal effect at $h=0$ is then
$$
Y_{t,0}(a_t) - Y_{t,0}(a_t') = \gamma(Y_{t-1},a_{t},W_t) - \gamma(Y_{t-1},a_{t}',W_t)
$$
and the causal effect at $h > 0$ is
$$
Y_{t,h}(a_t) - Y_{t,h}(a_t') = \gamma(Y_{t,h-1}(a_{t}),V_{t+h},W_{t+h})-\gamma(Y_{t,h-1}(a_{t}'),V_{t+h},W_{t+h}).
$$

If it exists, the marginal dynamic causal effect is, for $h \ge 0$,
\begin{align*}
\frac{\partial Y_{t,h}(a_t)}{\partial a_t} &=  \frac{\partial \gamma(Y_{t,h-1}(a_t),V_{t+h},W_{t+h})}{\partial Y_{t,h-1}(a_t)} 
\frac{\partial Y_{t,h-1}(a_t)}{\partial a_t} \\
&= \bigg(\prod_{j=1}^{h} \frac{\partial \gamma(Y_{t,j-1}(a_t),V_{t+j},W_{t+j})}{\partial Y_{t,j-1}(a_t)} \bigg)
\frac{\partial Y_{t,0}(a_t)}{\partial a_t} \\
&= \bigg(\prod_{j=1}^{h} \frac{\partial \gamma(Y_{t,j-1}(a_t),V_{t+j},W_{t+j})}{\partial Y_{t,j-1}(a_t)} \bigg)
\frac{\partial \gamma(Y_{t-1},a_{t},W_t)}{\partial a_t},
\end{align*}
which is typically stochastic.  

News impact causal studies appear in financial econometrics, but are typically not expressed in causal language, and instead discuss ``parameterized mechanisms.'' In that literature, a major topic is understanding how time-varying volatility of speculative assets (e.g., \cite%
{BollerslevEngleNelson(94)} and \cite{Shephard(05)}) change in response to news (e.g.,  \cite{CampbellHentschel(92)} and  
\cite{EngleNg(93)}). 

A further special case of this structure is the homogeneous Markov partially linear news impact {\tt PS}, which sets $Y_{t}(a_{1:t}) = \gamma Y_{t-1}(a_{1:t-1}) + \zeta(a_{t})$.  Then   
$$
Y_t = \gamma Y_{t-1} + \zeta(V_{t}),\quad 
Y_{t,0}(a_t) = \gamma Y_{t-1} + \zeta(a_{t}),
\quad Y_{t,h}(a_t) = \gamma Y_{t,h-1}(a_t) + \zeta(V_{t+h})
$$
for $h>0$. 
Thus, for $h\ge 0$, the dynamic causal effect is non-stochastic with  
$$
Y_{t,h}(a_t) - Y_{t,h}(a_t') = \gamma \{Y_{t,h-1}(a_t)- Y_{t,h-1}(a_t')\} = \gamma^h \{\zeta(a_{t}) - \zeta(a_{t}')\}.
$$

\begin{remark}[Slutzky-Frisch paradigm, continued]
In macroeconomics, it is often assumed that assignments of interest are observed, independent, ``exogenous'' sequences: assignments are ``shocks'' or ``impulses.'' To explore this, return to the homogeneous linear Markov potential system of Example \ref{ex:hlmps}, now with assignment mechanism 
\[ A_t = \Gamma V_t\]
and impose no features,
\[
\begin{array}{lll}
X_{t}(a_{1:t-1}) & = & 0 \\ 
Y_{t}(a_{1:t}) & = & \gamma_1 D_{t-1}(a_{1:t-2})+ \gamma_{0,A} a_t + \Omega W_{t}.
\end{array}%
\]%
The data is therefore shorter: $D_t=(A_t^{\mathtt{T}},Y_t^{\mathtt{T}})^{\mathtt{T}}$ and $\varepsilon_t = (V_t^{\mathtt{T}}, W_t^{\mathtt{T}})^{\mathtt{T}}$. Recall $\varepsilon_t$ is independent through time.  The DGP is thus  
$D_t = \phi D_{t-1} + B \varepsilon_t$ where now
$$ 
\phi = \bigg(\begin{array}{ll}  0 \\
 \gamma_1
\end{array}\bigg), \quad 
B = \bigg(\begin{array}{ll} \Gamma & 0 \\
\gamma_{0,A} \Gamma & \Omega
\end{array}\bigg).
$$
As before, the dynamic causal effect is $\Psi_h( a_{t}-a_{t}^{\prime })$, though now
$$ \Psi_h= \phi ^{h}\left( 
    \begin{array}{c}
    I \\ 
    \gamma_{0,A} 
    \end{array}%
    \right)$$
and the IRF for assignments is
$$
\Theta_h e_1 = \phi^h B e_1 = \phi^h\left( 
    \begin{array}{c}
    I \\ 
    \gamma_{0,A} 
    \end{array}%
    \right) \Gamma.
$$
If $\Gamma = I$, i.e., the assignment is the shock, then $\Theta_h e_1  = \Psi_h$, and so the dynamic causal effect of moving  $a_t' = 0$ to $a_t = 1$ is exactly the IRF for assignments.
\end{remark}

\subsection{$m$-order potential systems}\label{sect:order}

We may also consider examples of the {\tt PS} that further restrict the temporal impact of assignments on potential outcomes and features.

\begin{definition}[$m$-order {\tt PS}]\label{def:morder}
Assume a {\tt PS}. It is $m$-order causal if, for each $t$, 
\[
A_t = \alpha_t(D_{t-m:t-1}, X_t, V_t)
\]
and
\begin{equation*}
Z_{t}(a_{1:t-m-1},a_{t-m:t})=Z_{t}(a_{1:t-m-1}^{\prime },a_{t-m:t}),\quad 
\text{for all}\quad a_{1:t},a_{1:t-m-1}^{\prime },\quad m\geq 0.
\end{equation*}
\end{definition}

For $m$-th order causal {\tt PS} we write the time $t$
potential outcome and feature using the shorthand $$Z_{t}(a_{t-m:t}),$$ burying the irrelevance of $a_{1:t-m-1}$. Again, this is a type of non-interference assumption \citep
{Cox(58book)}.

In the important case of Definition \ref{def:morder} where $\{\varepsilon_t\} = \{(U_t^\mathtt{T},V_t^\mathtt{T},W_t^\mathtt{T})^\mathtt{T}\}$ is a sequence of independent random vectors, then an $m$-order PS is $m$-order Markovian.  Hence statistical methods designed for $m$-order Markov processes can be used in this setting, but now they have causal content. 

\subsection{{\tt PS}-exogeneity}

The idea of exogeneity has a long history in econometrics and is defined in many different ways.  Some of time series literature on this topic is discussed in \cite{EngleHendryRichard(83)}. Here we give a definition in the context of a {\tt PS}, viewing exogeneity as a form of invariance with respect to a possible intervention.  That line of thought goes back at least to \cite{Simon(53)}. 
\begin{definition}[{\tt PS}-exogeneity]
    Assume a {\tt PS}.  An invariant coordinate of $Z_{t,h}(a_t)$ with respect to the intervention coordinate of the $a_t$ is labeled ``{\tt PS}-exogenous'' if this holds for all $t$ and $h$.  If all of $X_{t,h}(a_t)$ is invariant to all of $a_t$, then we call them ``{\tt PS}-exogenous features.'' 
\end{definition}

\begin{example}[Homogeneous Markov linear {\tt PS} with exogenous features] Return to the homogeneous Markov linear {\tt PS} from Example \ref{ex:hlmps}, but now constrain the dynamics of the feature such that
\[
\begin{array}{lll}
X_{t}(a_{1:t-1}) =  \chi_1 X_{t-1}(a_{1:t-2}) + \Delta U_{t}. 
\end{array}
\]%
Then $X_{t,0}(a_t) = \chi_1 X_{t-1} + \Delta U_{t}$, which does not depend upon $a_t$.  Further, $X_{t,h}(a_t) =  \chi_1 X_{t,h-1}(a_t) + \Delta U_{t+h}$, so $X_{t,h}(a_t)$ is invariant to $a_t$ for all $h$.  So the feature is {\tt PS}-exogenous with respect to $a_t$.    
\end{example} 

{\tt PS}-exogeneity is an important condition for extending the potential system to applications in, e.g., design-based causal inference, discussed further in Section \ref{sec:dbci}. 

\begin{example}[{\tt PS}-proxy]\label{ex:proxy} Sometimes researchers are interested in the causal effect of assignments on outcomes, but the assignments themselves are measured with error or only partially revealed (\cite{StockWatson(18)} discuss the relevant literature and associated linear methods).  Here we provide a nonparametric version of this setup.  Assume (i) a {\tt PS}; (ii) that $a_t$ splits as
$$
a_t = (a_t^{*\mathtt{T}},\bar{a}_t^{\mathtt{T}})^{\mathtt{T}},\quad  \mathcal{A}_t = \mathcal{A}_t^* \times \bar{\mathcal{A}}_t,\quad a_t^* \in \mathcal{A}^*_t, \ \bar{a}_t \in \bar{\mathcal{A}}_t;
$$
(iii) the entire $Z_{t,h}(a_t)$ is {\tt PS}-exogenous with respect to $\bar{a}_t$; (iv) write $D^*_t := (X_t^{\mathtt{T}},A_t^{*\mathtt{T}},Y_t^{\mathtt{T}})^{\mathtt{T}}$, split $V_t = (V_t^{*\mathtt{T}},\bar{V}_t^{\mathtt{T}})^{\mathtt{T}}$ and assume that the function $\alpha_t$ has the triangular form  
\begin{align*}
A_t^* &= \alpha^*_t(D^*_{1:t-1},X_t,V_t^*), \\
\bar{A_t} & = \bar{\alpha}_t(D^*_{1:t-1},X_t,\bar{V}_t,A_t^*); 
\end{align*}
(v) the $\bar{D}_t:=(X_t^\mathtt{T},\bar{A}_t^\mathtt{T},Y_t^\mathtt{T})^\mathtt{T}$ is observed for $t=1,...,T$; (vi) $A_t^*$ is not directly observed.  Then $\bar{A}_t$ is a {\tt PS}-proxy for the assignment $A_t^*$. An example of this is: 
$$
\bar{A}_t = a + B A_t^* + \bar{V}_t,\quad \bar{V}_t \ind A_t^*,\quad \E[\bar{V}_t]=0,\quad t=1,...,T, 
$$
where $a,B$ are non-stochastic. 
Here all the causal content in this model is entirely driven by $A_t^*$ but we only see a noisy (and possibly smaller or larger dimensional) version $\bar{A}_t$.  
\end{example}

\section{From predictions to causality}\label{Section: pred to causal}

With definitions and measures of dynamic causality now in place, we investigate assumptions and results that allow data-based predictions to have nonparametric causal interpretations.

\subsection{Assumptions on the sequential assignment mechanism}

A major way of progressing from data-based predictions to causality is to make assumptions that constrain the behavior of the {\tt SAM} from Assumption {\tt DGP.2}. To do this compactly we use the notation
$$
Y_{t,0:H}(\mathcal{A}_t) = \{Y_{t,h}(a_{t}): a_t \in \mathcal{A}_t, h =0,1,...,H\},
$$
collecting the $a_t$-potential branch at different time horizons.  

\begin{definition}[Constraints on the {\tt SAM}]\label{defn:AM} Assume a {\tt PS}.  
\begin{enumerate}
    \item If $\big[A_{t} \ind  Y_{t,0:H}(\mathcal{A}_t)\big] \mid D_{1:t-1},X_{t}$, we say the {\tt SAM} obeys ``branch-sequential unconfoundedness'' ({\tt SAM.BSU}). 
    \item If
    $
    \big[ A_{t} \ind Y_{t,0:H}(\mathcal{A}_t) \big] \mid D_{1:t-1},
    $
    we say the {\tt SAM} obeys ``branch-sequential randomization'' ({\tt SAM.BSR}). 
    \item If
    $
    \big[ A_{t} \ind Y_{t,0:H}(\mathcal{A}_t) \big] \mid X_t,
    $
    we say the {\tt SAM} obeys ``branch-unconfoundedness'' ({\tt SAM.BU}). 
    \item If
    $ A_{t} \ind Y_{t,0:H}(\mathcal{A}_t),
    $
    we say the {\tt SAM} obeys ``branch-randomization'' ({\tt SAM.BR}).   
    \end{enumerate} 
\end{definition}

Under special cases of the potential system, the conditions {\tt SAM.BSU}, {\tt SAM.BSR}, {\tt SAM.BU} and {\tt SAM.BR} follow under more primitive conditions. 

\begin{restatable}{theorem}{thmbranchseq}\label{thm:brancseq}  Assume the SEM {\tt PS} from Example \ref{ex:npsem}.

\begin{enumerate}
\item If   
$[V_t \ind W_t] \mid D_{1:t-1},X_t$ then {\tt SAM.BSU} holds.
\item If $X_t$ is $D_{1:t-1}$-measurable and $[V_t \ind W_t] \mid D_{1:t-1}$ then {\tt SAM.BSR} holds. 

\item If $A_t = V_t$ (assignments are independent through time) and $\big[V_t \ind (D_{1:t-1}, W_t)\big] \mid X_t$ then {\tt SAM.BU} holds.

\item If $A_t = V_t$ (assignments are independent through time) and $V_t \ind (D_{1:t-1},X_t, W_t)$ then {\tt SAM.BR} holds.
\end{enumerate}
\end{restatable}
\begin{proof} See the Appendix.\end{proof}

\begin{remark}[``Exogenous'' assignment noise]
By Theorem \ref{thm:brancseq}, using the decomposition, weak union, and contraction properties of conditional independence, and recalling that $X_t = \chi_t(D_{1:t-1},U_t)$, for {\tt SAM.BSU} to hold it is sufficient that we assume  $$V_t \ind (U_t, W_t)\mid D_{1:t-1}.$$ Similar conclusions hold for the other parts of Theorem \ref{thm:brancseq} based on similar assumptions about the structural noise terms.
\end{remark}

\begin{remark}[Independent assignments]
The extra condition that $\{A_t\}_{t>0}$ is a sequence of independent random variables, which appears for {\tt SAM.BU} and {\tt SAM.BR}, is  certainly strong.  It starred in Example \ref{ex:news} about the homogeneous Markov news impact {\tt PS}.  Much of the economic time series literature measures dynamic causal quantities through impulse response functions by assuming assignments are independent through time. (In linear models, independence assumptions are often replaced by martingale differences or weak white noise assumptions.)  The outcome process is still flexible.  Only the assignments are highly constrained. 
\end{remark}

Assumption {\tt SAM.BSU} is stated as the conditional independence of $A_{t}$ and all the elements of $\{Y_{t,h}(a_{t}): a_t \in \mathcal{A}_t, h =0,1,...,H\}$.  A formally weaker alternative condition is to require many pairs of conditional independence rather than a single very large joint conditional independence.  This kind of pairwise assumption appears often in cross-sectional and panel population-based inference, for example \cite{HernanRobins(18)}. We define such conditions below.

\begin{definition}[Pairwise constraints on the {\tt SAM}]\label{defn:AMpair} Assume a {\tt PS}.  
\begin{enumerate}
    \item If $ 
    [A_{t} \ind Y_{t,h}(a_{t})] \mid (X_t,D_{1:t-1})
    $
    for each $a_{t} \in \mathcal{A}_{t}, h=0,1,...,H$, we say the {\tt SAM} obeys {\tt SAM.BSU-}. 
    \item If
       $
    [A_{t} \ind Y_{t,h}(a_{t})] \mid D_{1:t-1}
    $
    for each $a_{t} \in \mathcal{A}_{t}, h=0,1,...,H$, we say the {\tt SAM} obeys {\tt SAM.BSR-}. 
    \item If
    $
    [A_{t} \ind Y_{t,h}(a_{t})] \mid X_t
    $
    for each $a_{t} \in \mathcal{A}_{t}, h=0,1,...,H$, we say the {\tt SAM} obeys {\tt SAM.BU-}. 
    \item If
        $ 
    A_{t} \ind Y_{t,h}(a_{t})
    $ 
    for each $a_{t} \in \mathcal{A}_{t}, h=0,1,...,H$, we say the {\tt SAM} obeys {\tt SAM.BR-}.   
    \end{enumerate} 
\end{definition}

Under Assumption {\tt SAM.BSU-}, Figure \ref{fig:swit_bsu} depicts the potential branch of a {\tt PS} as part of a Single World Intervention Template (SWIT), which is a concise representation of a set of Single World Intervention Graphs (SWIGs).  The SWIG framework takes DAGs (Directed Acyclic Graphs) as inputs and ``splits'' them into SWIGs at nodes being intervened on, in this case at the node that represents $A_t$. Nodes downstream of the intervention become potential branches (the $Y_{t,0}(a_t)$ and $D_{t,h}(a_t)$ for $h=0,1,...,H$). (For other properties of SWIGs, see \cite{RichardsonRobins(2013)}.) At a glance, Figure \ref{fig:swit_bsu} shows that conditioning on $(X_t, D_{1:t-1})$ (which are acting as observed confounders) makes $A_t$ independent of branch potential outcomes (per standard analysis of probabilistic graphical models, blocking outgoing arrows from $(X_t, D_{1:t-1})$ separates $A_t$ from all other nodes), which is exactly what Assumption {\tt SAM.BSU-} states.

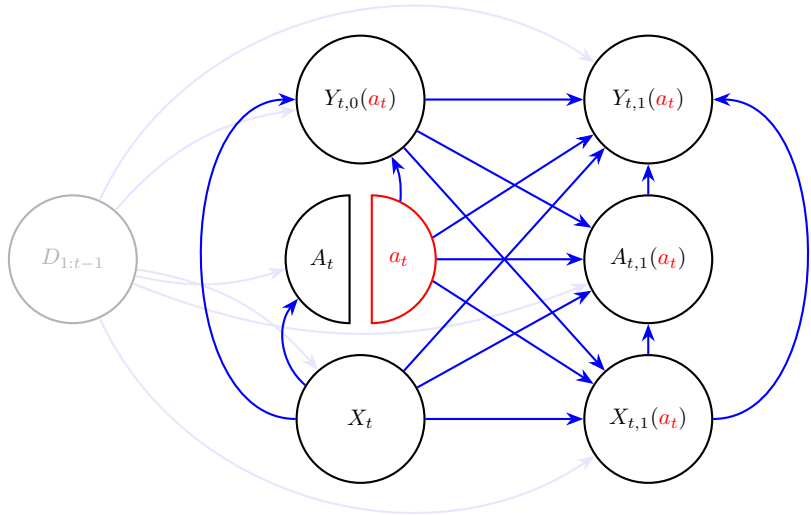
\begin{figure}[h!]
    \centering

\begin{tikzpicture}[scale=0.85, transform shape,
    normalnode/.style={
        draw,
        circle,
        thick,                       
        minimum size=2cm,          
        inner sep=0pt,
        align=center
    },
    xnode/.style={
        normalnode,
    },
    semicircle/.style={
        draw,
        shape=semicircle,
        thick,                      
        minimum width=2cm,        
        minimum height=1cm,      
        inner sep=0pt
    },
    redsemicircle/.style={
        draw,
        shape=semicircle,
        thick,                      
        minimum width=2cm,        
        minimum height=1cm,      
        inner sep=0pt,
        red
    },
    arrow/.style={
        ->,
        thick,
        blue,
        >=Stealth
    }
]

\node[normalnode] (yt) at (0, 2.5) {$Y_{t,0}(\textcolor{red}{a_t})$};
\node[normalnode] (ytp1) at (4.5, 2.5) {$Y_{t,1}(\textcolor{red}{a_t})$};
\node[xnode] (xt) at (0, -2.5) {$X_t$};
\node[xnode] (xtp1) at (4.5, -2.5) {$X_{t,1}(\textcolor{red}{a_t})$};

\node[normalnode] (atp1) at (4.5, 0) {$A_{t,1}(\textcolor{red}{a_t})$};

\node [
    draw, 
    circle, 
    color=black!30,
    minimum size=2cm, 
    inner sep=0pt, 
    thick
] (Q) at (-4.5, 0) {$D_{1:t-1}$};

\def\hshift{0.6} 

\node[semicircle, rotate=90] (at_left) at (-\hshift, 0) {\rotatebox{-90}{$A_t$}};
\node[redsemicircle, rotate=-90] (at_right) at (\hshift, 0) {\rotatebox{90}{$\textcolor{red}{a_t}$}};


\draw[arrow, color=blue!10, bend left=20] (Q) to (yt);
\draw[arrow, color=blue!10, bend left=50] (Q) to (ytp1);
\draw[arrow, color=blue!10, bend right=15] (Q) to (at_left);
\draw[arrow, color=blue!10, bend right=22] (Q) to (atp1);
\draw[arrow, color=blue!10, bend left=20] (Q) to (xt);
\draw[arrow, color=blue!10, bend right=50] (Q) to (xtp1);

\draw[arrow, bend right=15] (at_right) to (yt);
\draw[arrow] (at_right) -- (ytp1);
\draw[arrow] (at_right) -- (xtp1);
\draw[arrow] (at_right) -- (atp1);

\draw[arrow] (yt) -- (ytp1);
\draw[arrow] (yt) -- (atp1);
\draw[arrow] (yt) -- (xtp1);

\draw[arrow] (xt) -- (ytp1);
\draw[arrow] (xt) -- (atp1);
\draw[arrow] (xt) -- (xtp1);

\draw[arrow] (atp1) -- (ytp1);

\draw[arrow, bend left=90] (xt) to (yt);
\draw[arrow, bend left=45] (xt) to (at_left);
\draw[arrow, bend right=90] (xtp1) to (ytp1);
\draw[arrow] (xtp1) to (atp1);

\end{tikzpicture}
    
\caption{The {\tt PS} drawn as a ``Single World Intervention Template'' (SWIT), under branch-sequential unconfoundedness --- the {\tt SAM.BSU-} condition. }
    \label{fig:swit_bsu}
\end{figure}

\subsection{Identification of causal effects through predictions}

Using the conditions introduced in Definition \ref{defn:AM}, the following Theorem \ref{thm:bsu} shows that the causal summaries introduced in Definition \ref{defn:causaleffect} can be expressed in terms of population-based predictive quantities, delivering a version of the promise at the start of this paper: providing conditions where the difference of two data-based predictions are causal at horizon $h\ge0$.

\begin{restatable}{theorem}{thmbsu}\label{thm:bsu} Always assume a {\tt PS} in $L^1$ and set $h \ge 0$.   
\begin{enumerate}
    \item Additionally assume {\tt SAM.BR-}, then  
    $$
    {\tt ATE}_{t,h}(a_t,a_t')=\E[Y_{t+h}\mid A_{t}=a_{t}] 
    - \E[Y_{t+h}\mid A_{t}=a_{t}'].
    $$
    \item Additionally assume {\tt SAM.BU-}, then  
    $$
    {\tt CATE}_{t,h}(a_t,a_t') = \E[Y_{t+h}\mid X_t,A_{t}=a_{t}] 
    - \E[Y_{t+h}\mid X_t,A_{t}=a_{t}'].
    $$
    \item Additionally assume {\tt SAM.BSR-}, then  
    $$
    {\tt FTE}_{t,h}(a_t,a_t') = \E[Y_{t+h}\mid D_{1:t-1},A_{t}=a_{t}] 
    - \E[Y_{t+h}\mid D_{1:t-1},A_{t}=a_{t}'].
    $$
    \item Additionally assume {\tt SAM.BSU-}, then  
    $$
    {\tt CFTE}_{t,h}(a_t,a_t') = \E[Y_{t+h}\mid X_t,D_{1:t-1},A_{t}=a_{t}] 
    - \E[Y_{t+h}\mid X_t,D_{1:t-1},A_{t}=a_{t}'].
    $$
\end{enumerate}
\end{restatable}
\begin{proof} See the Appendix.\end{proof}

Such predictive quantities are not, in general, easy to estimate or approximate in practice.  However, we have made progress, moving from counterfactuals to observables which can be modeled and predicted: we have ``identified'' the causal objects defined in the earlier sections.  

Assumption {\tt SAM.BSR-} yields the SWIT in Figure \ref{fig:swit_BSR}, which has pruned off the confounders compared to Figure \ref{fig:swit_bsu}. The corresponding result relating causal quantities to data quantities under {\tt SAM.BSR-} is given in the third part of Theorem \ref{thm:bsu}. Assumption {\tt SAM.BU-} yields the SWIT given in the left-hand side of Figure \ref{fig:swit_other}. This is the same as the SWIT for {\tt SAM.BSU-} except the dependence on the history is removed.  The corresponding result relating causal quantities to data quantities under {\tt SAM.BU-} is given in the second part of Theorem \ref{thm:bsu}. Assumption {\tt  AM.BR-} yields the SWIT for branch-randomization, shown on the right-hand side of Figure \ref{fig:swit_other}.  Confounders no longer appear in the graph.  The corresponding result relating causal quantities to data quantities under {\tt SAM.BR-} is given in the first part of Theorem \ref{thm:bsu}. 

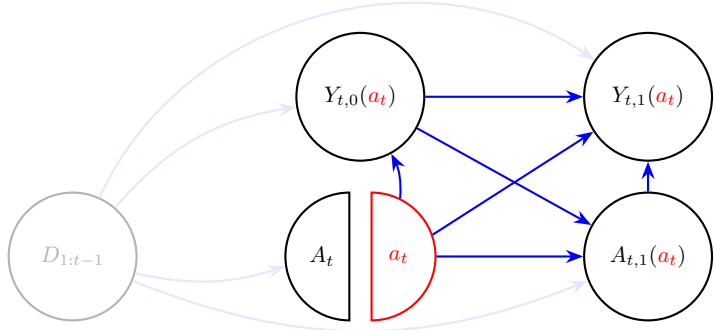
\begin{figure}[h!]
    \centering

\begin{tikzpicture}[scale=0.85, transform shape,
    normalnode/.style={
        draw,
        circle,
        thick,                       
        minimum size=2cm,          
        inner sep=0pt,
        align=center
    },
    xnode/.style={
        normalnode,
    },
    semicircle/.style={
        draw,
        shape=semicircle,
        thick,                      
        minimum width=2cm,        
        minimum height=1cm,      
        inner sep=0pt
    },
    redsemicircle/.style={
        draw,
        shape=semicircle,
        thick,                      
        minimum width=2cm,        
        minimum height=1cm,      
        inner sep=0pt,
        red
    },
    arrow/.style={
        ->,
        thick,
        blue,
        >=Stealth
    }
]

\node[normalnode] (yt) at (0, 2.5) {$Y_{t,0}(\textcolor{red}{a_t})$};
\node[normalnode] (ytp1) at (4.5, 2.5) {$Y_{t,1}(\textcolor{red}{a_t})$};

\node[normalnode] (atp1) at (4.5, 0) {$A_{t,1}(\textcolor{red}{a_t})$};

\node [
    draw, 
    circle, 
    color=black!30,
    minimum size=2cm, 
    inner sep=0pt, 
    thick
] (Q) at (-4.5, 0) {$D_{1:t-1}$};

\def\hshift{0.6} 

\node[semicircle, rotate=90] (at_left) at (-\hshift, 0) {\rotatebox{-90}{$A_t$}};
\node[redsemicircle, rotate=-90] (at_right) at (\hshift, 0) {\rotatebox{90}{$\textcolor{red}{a_t}$}};


\draw[arrow, color=blue!10, bend left=20] (Q) to (yt);
\draw[arrow, color=blue!10, bend left=50] (Q) to (ytp1);
\draw[arrow, color=blue!10, bend right=15] (Q) to (at_left);
\draw[arrow, color=blue!10, bend right=22] (Q) to (atp1);

\draw[arrow, bend right=15] (at_right) to (yt);
\draw[arrow] (at_right) -- (ytp1);
\draw[arrow] (at_right) -- (atp1);

\draw[arrow] (yt) -- (ytp1);
\draw[arrow] (yt) -- (atp1);

\draw[arrow] (atp1) -- (ytp1);

\end{tikzpicture}
    
\caption{The {\tt PS} drawn as a ``Single World Intervention Template'' (SWIT) under sequential randomization  --- the {\tt SAM.BSR-} case.}
    \label{fig:swit_BSR}
\end{figure}
 
\begin{figure}[h!]
  \centering
  \begin{minipage}[b]{0.4\textwidth}
    \flushleft
    \begin{tikzpicture}[scale=0.85, transform shape,
        normalnode/.style={
            draw,
            circle,
            thick,
            minimum size=2cm,
            inner sep=0pt,
            align=center
        },
        xnode/.style={normalnode},
        semicircle/.style={
            draw, shape=semicircle, thick,
            minimum width=2cm, minimum height=1cm,
            inner sep=0pt
        },
        redsemicircle/.style={
            draw, shape=semicircle, thick,
            minimum width=2cm, minimum height=1cm,
            inner sep=0pt, red
        },
        arrow/.style={->, thick, blue, >=Stealth}
    ]
    \node[normalnode] (yt)   at (0,  2.5) {$Y_{t,0}(\textcolor{red}{a_t})$};
    \node[normalnode] (ytp1) at (4.5,2.5) {$Y_{t,1}(\textcolor{red}{a_t})$};
    \node[xnode]      (xt)   at (0, -2.5) {$X_t$};
    \node[xnode]      (xtp1) at (4.5,-2.5) {$X_{t,1}(\textcolor{red}{a_t})$};
    \node[normalnode] (atp1) at (4.5, 0)   {$A_{t,1}(\textcolor{red}{a_t})$};
    \def\hshift{0.6}
    \node[semicircle, rotate=90]    (at_left)  at (-\hshift,0) {\rotatebox{-90}{$A_t$}};
    \node[redsemicircle, rotate=-90] (at_right) at (\hshift,0)  {\rotatebox{90}{$\textcolor{red}{a_t}$}};
    \draw[arrow, bend right=15] (at_right) to (yt);
    \draw[arrow] (at_right) -- (ytp1);
    \draw[arrow] (at_right) -- (xtp1);
    \draw[arrow] (at_right) -- (atp1);
    \draw[arrow] (yt) -- (ytp1);
    \draw[arrow] (yt) -- (atp1);
    \draw[arrow] (yt) -- (xtp1);
    \draw[arrow] (xt) -- (ytp1);
    \draw[arrow] (xt) -- (atp1);
    \draw[arrow] (xt) -- (xtp1);
    \draw[arrow] (atp1) -- (ytp1);
    \draw[arrow, bend left=90]  (xt)    to (yt);
    \draw[arrow, bend left=45]  (xt)    to (at_left);
    \draw[arrow, bend right=90] (xtp1) to (ytp1);
    \draw[arrow] (xtp1) -- (atp1);
    \end{tikzpicture}
  \end{minipage}\hfill
  \begin{minipage}[b]{0.4\textwidth}
    \flushleft
    \begin{tikzpicture}[scale=0.85, transform shape,
        normalnode/.style={
            draw,
            circle,
            thick,
            minimum size=2cm,
            inner sep=0pt,
            align=center
        },
        semicircle/.style={
            draw, shape=semicircle, thick,
            minimum width=2cm, minimum height=1cm,
            inner sep=0pt
        },
        redsemicircle/.style={
            draw, shape=semicircle, thick,
            minimum width=2cm, minimum height=1cm,
            inner sep=0pt, red
        },
        arrow/.style={->, thick, blue, >=Stealth}
    ]
    \node[normalnode] (yt)   at (0,  2.5) {$Y_{t,0}(\textcolor{red}{a_t})$};
    \node[normalnode] (ytp1) at (4.5,2.5) {$Y_{t,1}(\textcolor{red}{a_t})$};
    \node[normalnode] (atp1) at (4.5, 0)   {$A_{t,1}(\textcolor{red}{a_t})$};
    \def\hshift{0.6}
    \node[semicircle, rotate=90]    (at_left)  at (-\hshift,0) {\rotatebox{-90}{$A_t$}};
    \node[redsemicircle, rotate=-90] (at_right) at (\hshift,0)  {\rotatebox{90}{$\textcolor{red}{a_t}$}};
    \draw[arrow, bend right=15] (at_right) to (yt);
    \draw[arrow] (at_right) -- (ytp1);
    \draw[arrow] (at_right) -- (atp1);
    \draw[arrow] (yt) -- (ytp1);
    \draw[arrow] (yt) -- (atp1);
    \draw[arrow] (atp1) -- (ytp1);
    \end{tikzpicture}
  \end{minipage}
\caption{The {\tt PS} drawn as a ``Single World Intervention Template'' (SWIT).  Left-hand side shows the branch unconfoundedness --- the {\tt SAM.BU-} case. The right-hand side shows the the branch randomization --- the {\tt SAM.BR-} case. }
\label{fig:swit_other}
\end{figure}
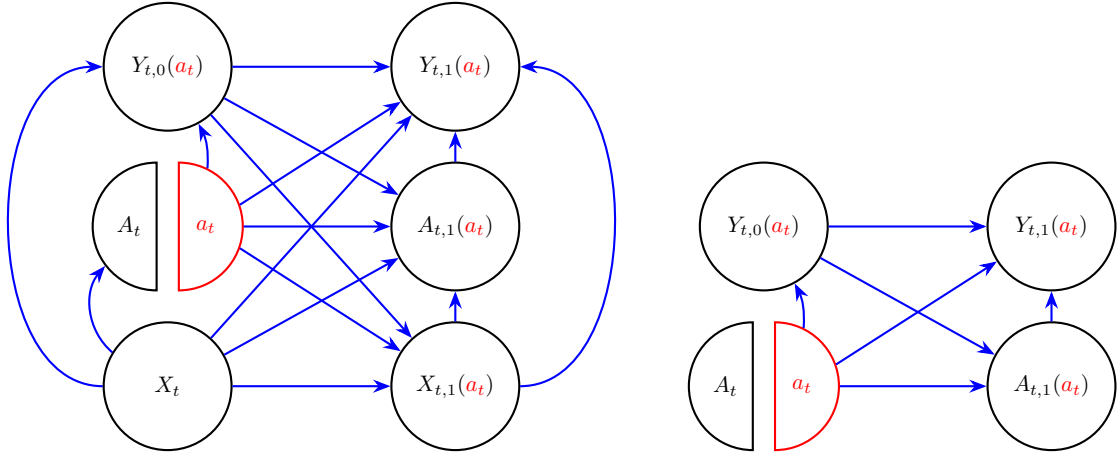

\subsection{Linear projection and covariance stationarity}

We saw in Theorem \ref{thm:bsu} various conditions on the {\tt PS} that allow the ${\tt ATE}_{t,h}(a_t,a_t')$, ${\tt CATE}_{t,h}(a_t,a_t')$, ${\tt FTE}_{t,h}(a_t,a_t')$ and ${\tt CFTE}_{t,h}(a_t,a_t')$ to be written as data-based (conditional) expectations. In some applications it is helpful to replace the data-based conditional expectations by projections. Economists often use linear projections (see, e.g., \cite{AngristPischke(09)}), which we now consider in detail.  

Before proceeding, to establish notation,  denote the usual linear projection of a generic random variable $A$ on 1 and a generic random variable $B$ as $\operatorname{LP}[A\mid 1,B] = \kappa + \beta B$, where $\beta=\cov(A,B)\var(B)^{-1}$ and $\kappa = \E[A] - \beta \E[B]$, so that 
\begin{align*}
(\kappa,\beta) &= \underset{k,b}{\arg} \min \E_{A,B}[(A-k - b B)^2] \\
 & = \underset{k,b}{\arg} \min \E_B[(\E[A|B] -k - b B)^2]. 
\end{align*} 
For this setup to make sense, $A,B$ must both be in $L^2$ and $\var(B)>0$. Here $\E_{A,B}$ is the expectation with respect to both $A$ and $B$. $\E_{B}$ is the expectation solely with respect to $B$.  

Assume the {\tt PS} is in $L^2$ and work under the {\tt SAM.BR-} condition, so $$\E[Y_{t+h}(a_t)] = \E[Y_{t+h} \mid A_t=a_t].$$ Then the linear projection of $Y_{t+h}$ on 1, $A_t$ is  
$$
\operatorname{LP}[Y_{t+h} \mid 1,A_t=a_t] = \kappa_{t,h} + \beta_{t,h} a_t,\quad \text{where} \quad \beta_{t,h}  = 
\cov(Y_{t+h},A_t)   
\var(A_t)^{-1}
$$
assuming $\var(A_t)>0$.  
So the linear projection of the conditional average treatment effect ${\tt ATE}_{t,h}(a_t,a_t') = \E[Y_{t+h}(a_t)] - \E[Y_{t+h}(a_t')]$ (which is non-stochastic) is
$$
\operatorname{LP}[Y_{t+h} \mid 1,A_t=a_t] - \operatorname{LP}[Y_{t+h} \mid 1,A_t=a_t'] = \beta_{t,h}(a_t - a_t'),
$$
(which is non-stochastic). 

\begin{remark}[Linear projection and binary assignment] If the {\tt PS} is in $L^2$ and $\mathcal{A}_t = \{0,1\}$, then 
$$
\E[Y_{t+h}|A_t=1] - \E[Y_{t+h}|A_t=0],
$$
the ${\tt ATE}_{t,h}(1,0)$ under {\tt SAM.BR-}, is just a difference in means, and can be implemented exactly with a linear projection, as
\begin{align*}
    \beta_{t,h} &= 
    \cov(Y_{t+h},A_t)   
\var(A_t)^{-1} \\&= (\E[Y_{t+h}A_t] -  \E[Y_{t+h}]E[A_t]) \{\E[A_t](1-E[A_t])\}^{-1} \\
&= \E[Y_{t+h}|A_t=1] - \E[Y_{t+h}|A_t=0].
\end{align*} 
Further assuming that the {\tt PS} is covariance stationary then makes estimation easy.
\end{remark}

\begin{remark}[Linear projection and noisy assignment]
Recall the {\tt PS}-proxy setting from Example \ref{ex:proxy}, and assume that all the causal content in the model is entirely driven by scalar $A_t^*$, but we only see a noisy version $\bar{A}_t$, given by
$$
\bar{A}_t = a + B A_t^* + \bar{V}_t,\quad \bar{V}_t \ind D_{1:T}^*,\quad \E[\bar{V}_t]=0,\quad t=1,...,T, 
$$
where $a,B$ are non-stochastic. Notice that in this setting, the linear projection coefficient is
\begin{align*}
     \cov(Y_{t+h},\bar A_t)   
\var(\bar A_t)^{-1} = \beta_{t,h} \left(  \frac{B\var (A_t^*)}{B^2 \var(A_t^*) + \var(\bar{V}_t)} \right).
\end{align*}
When $B=1$, this recovers the well-known consequence of linear regression with classical measurement error in the regressor: attenuation bias. However, if $B\neq 1$, even if $\var(\bar{V}_t)=0$, notice that bias may or may not be attenuating, and depends on the value of $B$.
\end{remark}

We now turn to a subtler case.  Assume the {\tt PS} is in $L^2$ and work under the {\tt SAM.BU-} condition, so $$\E[Y_{t+h}(a_t) \mid X_t] = \E[Y_{t+h} \mid A_t=a_t,X_t].$$ Then the linear projection of $Y_{t+h}$ on 1, $A_t$ and $X_t$ is  
$$
\operatorname{LP}[Y_{t+h} \mid 1,A_t=a_t,X_t] = \kappa_{t,h} + \beta_{t,h} a_t + \delta_{t,h}X_t,
$$
where 
\begin{align*}
\begin{pmatrix}
\beta_{t,h} \\
\delta_{t,h}
\end{pmatrix} = 
\cov(Y_{t+h},(A_t^\mathtt{T},X_t^\mathtt{T})^\mathtt{T})   
\begin{pmatrix}
\var(A_t) & \cov(A_t,X_t)\\
\cov(X_t,A_t)& \var(X_t)
\end{pmatrix}^{-1} 
\end{align*}
assuming $\var((A_t^\mathtt{T},X_t^\mathtt{T})^\mathtt{T})>0$.  
So the linear projection of the conditional average treatment effect ${\tt CATE}_{t,h}(a_t,a_t') = \E[Y_{t+h}(a_t)\mid X_t] - \E[Y_{t+h}(a_t')\mid X_t]$ (which is stochastic) is
$$
\operatorname{LP}[Y_{t+h} \mid 1,A_t=a_t,X_t] - \operatorname{LP}[Y_{t+h} \mid 1,A_t=a_t',X_t] = \beta_{t,h}(a_t - a_t')
$$
(which is non-stochastic, but the first two moments of $X_t$ influence $\beta_{t,h}$). 

Under covariance stationarity of $\{D_{t}\}$, one may write $\beta_{t,h}:=\beta_h$ for all $t$.  A weaker assumption is to assume $\{D_{t}\}$ is only locally covariance stationary, where the dependence through time changes slowly (e.g., \cite{Dahlhaus(12)}).

Covariance stationarity of $\{D_{t}\}$ is a sufficient condition to produce a time-invariant linear projected causal effect, but it is not necessary. By the Frisch-Waugh-Lovell theorem \citep{Yule(1907)}, we may also write
\[
\beta_{t,h} = \cov(Y_{t+h},A_t^\perp) \var(A_t^\perp)^{-1}
\]
for $A_t^\perp := A_t - \operatorname{LP}[A_t \mid 1, X_t]$. As such, we can see it is also sufficient to only require that $\{(Y_{t}, A_t^\perp)\}$ is covariance stationary to yield $\beta_{t,h}:=\beta_h$ for all $t$. Moreover, for any random $B_{t+h}$ such that $\cov(B_{t+h}, A_t^\perp)=0$, we have that
\[
\beta_{t,h} = \cov(Y_{t+h}-B_{t+h},A_t^\perp)\var(A_t^\perp)^{-1},
\]
i.e., regressing $Y_{t+h}-B_{t+h}$ on $A_t^\perp$ yields the same $\beta_{t,h}$ as regressing $Y_{t+h}$ on $A_t^\perp$. Certain choices of $B_{t+h}$ may help improve precision in downstream estimation, or it may be more plausible that $\{(Y_{t+h}-B_{t+h},A_t^\perp)\}$ is covariance stationary. An example of this is where $\{Y_{t}\}$ is an integrated variable but $\{B_{t}\}$ is a detrender or synthetic control.

\begin{remark}[Local projection and economics]
In the econometric time series literature the linear projection approach in the context of dynamic causal observational studies is associated with \cite{Jorda(05)} under the heading ``local projection.'' It is typically stated in the context of a covariance stationary time series. 
Work on inferential aspects of local projection includes \cite{PlagborgMollerWolf(21)}, \cite{OleaPlagborgMoller(21)} and \cite{AdamekSmeekesWilms(24)}.  
\end{remark}

\begin{remark}[Conditional linear projection and binary assignment]
If the {\tt PS} is covariance stationary, is in $L^2$, and $\mathcal{A}_t = \{0,1\}$, we can consider the conditional linear projection 
$$
\operatorname{CLP}[Y_{t+h}\mid 1,X_t;A_t=a_t] = \kappa^{(a_t)}_{h} + \beta^{(a_t)}_{h}X_t, 
$$ 
where $\beta^{(a_t)}_{h} = \cov((Y_{t+h},X_t)|A_t=a_t) \{\var(X_t|A_t=a_t)\}^{-1}$ and $\kappa^{(a_t)}_{h} = \E[Y_{t+h}|A_t=a_t] - \beta^{(a_t)}_{h} \E[X_{t}|A_t=a_t]$, which is also quite easy to estimate.  Note that $\operatorname{CLP}[Y_{t+h}|1,X_t;A_t=1] - \operatorname{CLP}[Y_{t+h}|1,X_t;A_t=0]$ can then be expressed as
\begin{align*}
(\kappa^{(1)}_{h} - \kappa^{(0)}_{h})+ (\beta^{(1)}_{h} - \beta^{(0)}_{h})X_t
&= \E[Y_{t+h}|A_t=1] - \E[Y_{t+h}|A_t=0] \\
&+ \beta^{(1)}_{h}(X_t - \E[X_t|A_t=1])
- \beta^{(0)}_{h}(X_t - \E[X_t|A_t=0]),
\end{align*}
a \textit{conditional} linear projection of the ${\tt CATE}_h(1,0)$ under {\tt SAM.BU-}. This conditional linear projection of the ${\tt CATE}_h(1,0)$ is stochastic, just like the ${\tt CATE}_h(1,0)$ itself, allowing for heterogeneity in the causal effect summary for different realized values of $X_t$ (which is ruled out in the unconditional linear projection considered earlier). 

This conditional linear projection approach generalizes in multiple ways, e.g.: (i) where $\mathcal{A}_t$ has a finite number of atoms, not just two, and (ii) computing $\kappa_h^{(a_t)}$ and $\beta_h^{(a_t)}$ by kernels applied to a continuous $a_t \in \mathcal{A}_t$ (though still imposing linearity in $X_t$).     
\end{remark}

\subsection{Causal summaries and strict stationarity}

Assume throughout this subsection strict stationarity of $\{D_t\}$ (so the dimensions of features, assignments and outcomes are time invariant) and the {\tt PS} is in $L^1$. Further assume $\mathcal{A}_t = \mathcal{A}$ for all $t$, i.e., the assignment space is not changing over time.  

Under the condition {\tt SAM.BR-} the 
$$
\E[Y_{t+h}(a)] = \E[Y_{t+h}|A_t=a] = \mu_h(a),
$$
where $\mu_h = \{\mu_h(a):a\in \mathcal{A}\}$ is a deterministic function. Thus the ${\tt ATE}_h(a,a') =\mu_h(a)-\mu_h(a')$. Typically, for strictly stationary processes, $\mu_h$ would be estimated as a nonparametric regression of $Y_{t+h}$ on $A_t$, e.g., through Nadaraya-Watson kernel regression, local linear regressions \citep{FanYao(05)}, splines, or neural networks.    

Under the weaker condition {\tt SAM.BU-} the 
$$
\E[Y_{t+h}(a)|X_t] = \E[Y_{t+h}|A_t=a,X_t] = \mu_h(a,X_t),
$$
where $\mu_h = \{\mu_h(a,x):a \in \mathcal{A},x \in \mathcal{X}\}$ is a deterministic function. Thus  
$$
{\tt CATE}_h(a,a') =\mu_h(a,X_t)-\mu_h(a',X_t)$$
which implies that, by iterated expectations   
$${\tt ATE}_h(a,a') =E_{X_1}[\mu_h(a,X_1)]-E_{X_1}[\mu_h(a',X_1)].
$$
The $\mu_h$ can be estimated by a non-parametric regression of $Y_{t+h}$ on $A_t$ and $X_t$.   

Under the condition {\tt SAM.BSU-} and imposing the {\tt PS} is $m$-order Markovian (see Section \ref{sect:order}), then 
$$
\E[Y_{t+h}(a)|X_t,D_{t-m:t-1}] = \E[Y_{t+h}|A_t=a,X_t,D_{t-m:t-1}]=\mu_h(a,X_t,D_{t-m:t-1}),
$$
where $\mu_h = \{\mu_h(a,x,d):a\in \mathcal{A},x \in \mathcal{X}, d \in \mathcal{D}^m\}$ is a deterministic function. Thus  
$$
{\tt CFTE}_h(a,a') =\mu_h(a,X_t,D_{t-m:t-1})-\mu_h(a',X_t,D_{t-m:t-1})$$
which implies that, by iterated expectations  
$${\tt ATE}_h(a,a') =E_{X_{m+1},D_{1:m}}[\mu_h(a,X_{m+1},D_{1:m})]-E_{X_{m+1},D_{1:m}}[\mu_h(a',X_{m+1},D_{1:m})]. 
$$
The $\mu_h$ can again be estimated by a non-parametric regression of $Y_{t+h}$ on $A_t$, $X_t$ and $D_{t-m:t-1}$. 

Under the condition {\tt SAM.BSR-} plus imposing the {\tt PS} is $m$-order Markovian, then 
$$
\E[Y_{t+h}(a)|D_{t-m:t-1}] = \E[Y_{t+h}|A_t=a,D_{t-m:t-1}]=\mu_h(a,D_{t-m:t-1}),
$$
where $\mu_h = \{\mu_h(a,d):a\in \mathcal{A},d \in \mathcal{D}^m\}$ is a deterministic function. Thus  
$$
{\tt FTE}_h(a,a') =\mu_h(a,D_{t-m:t-1})-\mu_h(a',D_{t-m:t-1})$$
which implies that, by iterated expectations  
$${\tt ATE}_h(a,a') =E_{D_{1:m}}[\mu_h(a,D_{1:m})]-E_{D_{1:m}}[\mu_h(a',D_{1:m})]. 
$$
Here $\mu_h$ can be estimated by a non-parametric regression of $Y_{t+h}$ on $A_t$ and $D_{t-m:t-1}$.

\subsection{Influence curves and double robustness}

Assume that the {\tt PS} is in $L^1$, that $\mathcal{A}_t$ is made up of a finite number of atoms, and that {\tt SAM.BSU-} holds. Define $\lambda_{a_t}(X_t,D_{1:t-1}):=P(A_{t}=a_t \mid X_t,D_{1:t-1})$, the propensity score, and assume it is bounded away from zero and one for all $a_t \in \mathcal{A}_t$. Recall that
\[
{\tt CFTE}_{t,h}(a_t,a_t') = \E[Y_{t+h}\mid X_t,D_{1:t-1},A_{t}=a_{t}] 
    - \E[Y_{t+h}\mid X_t,D_{1:t-1},A_{t}=a_{t}'].
\]
As such we have that
\[
{\tt ATE}_{t,h}(a_t,a_t') =\E[{\tt CFTE}_{t,h}(a_t,a_t')] = \E[\E[Y_{t+h}\mid X_t,D_{1:t-1},A_{t}=a_{t}]] 
    - \E[\E[Y_{t+h}\mid X_t,D_{1:t-1},A_{t}=a_{t}']].
\]
From the literature on semiparametric inference, we know that this is a classic missing data functional \citep{Kennedy(2024)}, for which the influence curve (or ``efficient influence function'') in a fully nonparametric observed-data model for $\left(Y_{t+h}, A_t, X_t, D_{1: t-1}\right)$ is
\begin{align*}
    \mathbb{IF}({\tt ATE}_{t,h}(a_t,a_t')) &= \E[Y_{t+h}\mid X_t,D_{1:t-1},A_{t}=a_{t}] + \frac{1(A_t = a_t)}{\lambda_{a_t}(X_t,D_{1:t-1})} \big\{Y_{t+h} -   \E[Y_{t+h}\mid X_t,D_{1:t-1},A_{t}=a_{t}]\big\} \\
    &- \E[Y_{t+h}\mid X_t,D_{1:t-1},A_{t}=a_{t}'] - \frac{1(A_t = a_t')}{\lambda_{a_t'}(X_t,D_{1:t-1})} \big\{Y_{t+h} -   \E[Y_{t+h}\mid X_t,D_{1:t-1},A_{t}=a_{t}']\big\} \\
    &- {\tt ATE}_{t,h}(a_t,a_t').
\end{align*}

Influence curves are random objects that can be used to construct semiparametric efficient estimators, e.g., the doubly robust estimators that are familiar from \cite{RobinsRotnitzkyZhao(94)} and \cite{HernanRobins(18)}, or the double/debiased machine learning estimators familiar from \cite{Chernozhukov(2018)}.  Semiparametric efficient inference on nonparametrically defined impulse response functions is explored in part in, e.g., \cite{BallinariWehrli(24)}, building from \cite{RambachanShephard(21)}, and can be grounded in the {\tt PS}. The desirable properties of influence function-based estimators are discussed in many works (e.g.,  \cite{RobinsRotnitzkyZhao(94)} or \cite{Chernozhukov(2018)}, or see \cite{Kennedy(2024)} for an overview).

\section{Extensions}\label{sec:ext}

\subsection{Instrumental variables and local causal effect summaries}\label{sec:iv}

The {\tt PS} also accommodates identification of local summaries of causal effects using instrumental variables, in the spirit of \cite{ImbensAngrist(94)}. 

\begin{definition}[Instrumental variables {\tt PS}]
Assume the triangular nonparametric SEM {\tt PS} from Example \ref{ex:npsem}. The instrumental variables (IV) {\tt PS} further assumes that
\begin{align*}
    \chi_t(y_{1:t-1},a_{1:t-1},x_{1:t-1},u_{t}) &=\chi_t(y_{1:t-1},a_{1:t-1},x_{1:t-1}',u_{t}), \\
    \alpha_t(y_{1:t-1},a_{1:t-1},x_{1:t-1},x_t,v_{t}) &=\alpha_t(y_{1:t-1},a_{1:t-1},x_{1:t-1}',x_t,v_{t}), \\
    \gamma_t(y_{1:t-1},a_{1:t-1},x_{1:t-1},x_t,a_t,w_{t})&= \gamma_t(y_{1:t-1},a_{1:t-1},x'_{1:t-1},x_t',a_t,w_{t}),
\end{align*}
for all $x_{1:t},x_{1:t}' \in \mathcal{X}_{1:t}, a_{1:t} \in \mathcal{A}_{1:t}, y_{1:t-1} \in \mathcal{Y}_{1:t-1}, w_t \in \mathcal{W}_t, u_t \in \mathcal{U}_t, v_t \in \mathcal{V}_t$ for all $t$.  We write them in shorthand as 
$$
\chi_t(y_{1:t-1},a_{1:t-1},u_{t}),\quad 
\alpha_t(y_{1:t-1},a_{1:t-1},x_t,v_{t}), \quad 
\gamma_t(y_{1:t-1},a_{1:t},w_{t}).
$$
\end{definition}

The definition of the IV {\tt PS} imposes further (exclusion) restrictions on the causal relationships of the variables in the {\tt PS}. Under the IV {\tt PS}, we may write
\begin{eqnarray*}
X_t(a_{1:t-1}) &=& \chi_t(\tilde D_{1:t-1}(a_{1:t-1}),U_{t}),  \\
A_t   &=& \alpha_t(\tilde D_{1:t-1}, X_t,V_t), \\
Y_{t}(a_{1:t}) &=&\gamma_t(\tilde D_{1:t-1}(a_{1:t-1}),a_t,W_{t}),
\end{eqnarray*}
where $\tilde D_{t}:=\left(A_{t}^{\mathtt{T}},Y_{t}^{\mathtt{T}}\right)^{\mathtt{T}}$ and $\tilde D_{1:t-1}(a_{1:t-1}) := \{\tilde D_1(a_1),...,\tilde D_{t-1}(a_{1:t-1})\}$ where $\tilde D_t(a_{1:t}) := \{a_t,Y_t(a_{1:t})\}$. Figure \ref{fig:swit_iv} depicts the instrumental variables {\tt PS} in a SWIT.

Under this system definition, we have that feature $X_t$ is a valid \textit{instrumental variable} for the time $t$ assignment conditional on  $\tilde D_{1:t-1}$, so long as (sufficiently) the $U_t \ind (V_t, W_t) \mid \tilde D_{1:t-1}$ and an instrument relevance condition holds. By further making a monotonicity assumption familiar from \cite{ImbensAngrist(94)}, letting $\mathcal{A}_t = \mathcal{X}_t = \{0,1\}$ for all $t$, and, for all $h$, defining $A_t(x_t):=\alpha_t(\tilde D_{1:t-1}, x_t,V_t)$, we can identify a local summary of the causal effect,
\[
E[Y_{t,h}(1) - Y_{t,h}(0) \mid \mathbf{1}\{A_t(1) > A_t(0) \}=1,  \tilde D_{1:t-1}].
\]
The event $\{A_t(1) > A_t(0)\}$ can be thought of as the single unit in the time series ``complying'' with the instrument.

\begin{restatable}{theorem}{thmiv}\label{thm:iv}
Assume an instrumental variables {\tt PS} where $\mathcal{A}_t = \mathcal{X}_t = \{0,1\}$ for all $t$. Further assume: 
\begin{enumerate}
    \item[(i)] $[U_t \ind (V_t,W_t)] \mid \tilde{D}_{1:t-1}$.
    \item[(ii)] $\E[A_{t} \mid X_t = 1, \tilde{D}_{1:t-1}] - \E[A_{t} \mid X_t = 0, \tilde{D}_{1:t-1}] > 0$ almost surely.
    \item[(iii)] $A_t(1) \geq A_t(0)$ almost surely.
\end{enumerate}
Then, almost surely, for any $h = 0,1, ..., H$,
\[
\E[Y_{t,h}(1) - Y_{t,h}(0) \mid \mathbf{1}\{A_t(1) > A_t(0) \}=1,  \tilde{D}_{1:t-1}] = \frac{\E[Y_{t+h} \mid X_t = 1, \tilde{D}_{1:t-1}] - \E[Y_{t+h} \mid X_t = 0, \tilde{D}_{1:t-1}]}{\E[A_{t} \mid X_t = 1, \tilde{D}_{1:t-1}] - \E[A_{t} \mid X_t = 0, \tilde{D}_{1:t-1}]}.
\]
\end{restatable}
\begin{proof}
See the Appendix.
\end{proof}

The first condition in Theorem \ref{thm:iv}, in conjunction with the exclusion restrictions imposed by the IV {\tt PS}, grants that
\[
X_t \ind \big[A_t(1),A_t(0), Y_{t,h}(1), Y_{t,h}(0)\big] \mid \tilde D_{1:t-1}.
\]
The second and third conditions of Theorem \ref{thm:iv} mirror the relevance and monotonicity assumptions, respectively, introduced in \cite{ImbensAngrist(94)}.

The empirical setting represented by the IV {\tt PS} may be relevant to, e.g., a health system that wants to know the causal effect of ingesting a drug on a patient's health over time, but can only randomly encourage the patient to do so with a text reminder; or a ride-share application company that wants to understand the causal effect of augmenting some aspect of city-wide driver behavior on app engagement, but can only provide that city's drivers with randomized incentives to encourage desired behavior at scale.

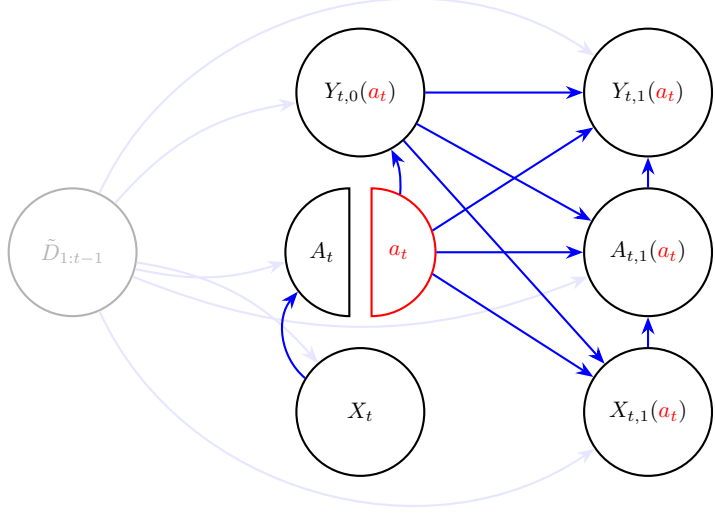
\begin{figure}[h!]
    \centering

\begin{tikzpicture}[scale=0.85, transform shape,
    normalnode/.style={
        draw,
        circle,
        thick,                       
        minimum size=2cm,            
        inner sep=0pt,
        align=center
    },
    xnode/.style={
        normalnode,
    },
    semicircle/.style={
        draw,
        shape=semicircle,
        thick,                      
        minimum width=2cm,          
        minimum height=1cm,         
        inner sep=0pt
    },
    redsemicircle/.style={
        draw,
        shape=semicircle,
        thick,                      
        minimum width=2cm,          
        minimum height=1cm,         
        inner sep=0pt,
        red
    },
    arrow/.style={
        ->,
        thick,
        blue,
        >=Stealth
    }
]

\node[normalnode] (yt) at (0, 2.5) {$Y_{t,0}(\textcolor{red}{a_t})$};
\node[normalnode] (ytp1) at (4.5, 2.5) {$Y_{t,1}(\textcolor{red}{a_t})$};
\node[xnode] (xt) at (0, -2.5) {$X_t$};
\node[xnode] (xtp1) at (4.5, -2.5) {$X_{t,1}(\textcolor{red}{a_t})$};

\node[normalnode] (atp1) at (4.5, 0) {$A_{t,1}(\textcolor{red}{a_t})$};

\node [
    draw, 
    circle, 
    color=black!30,
    minimum size=2cm, 
    inner sep=0pt, 
    thick
] (Q) at (-4.5, 0) {$\tilde{D}_{1:t-1}$};


\def\hshift{0.6} 

\node[semicircle, rotate=90] (at_left) at (-\hshift, 0) {\rotatebox{-90}{$A_t$}};
\node[redsemicircle, rotate=-90] (at_right) at (\hshift, 0) {\rotatebox{90}{$\textcolor{red}{a_t}$}};


\draw[arrow, color=blue!10, bend left=20] (Q) to (yt);
\draw[arrow, color=blue!10, bend left=50] (Q) to (ytp1);
\draw[arrow, color=blue!10, bend right=15] (Q) to (at_left);
\draw[arrow, color=blue!10, bend right=22] (Q) to (atp1);
\draw[arrow, color=blue!10, bend left=20] (Q) to (xt);
\draw[arrow, color=blue!10, bend right=50] (Q) to (xtp1);


\draw[arrow, bend right=15] (at_right) to (yt);
\draw[arrow] (at_right) -- (ytp1);
\draw[arrow] (at_right) -- (xtp1);
\draw[arrow] (at_right) -- (atp1);

\draw[arrow] (yt) -- (ytp1);
\draw[arrow] (yt) -- (atp1);
\draw[arrow] (yt) -- (xtp1);

\draw[arrow] (atp1) -- (ytp1);

\draw[arrow, bend left=45] (xt) to (at_left);
\draw[arrow] (xtp1) to (atp1);

\end{tikzpicture}

\caption{The IV {\tt PS} drawn as a ``Single World Intervention Template'' (SWIT) where $X_t$ is an instrumental variable.}
    \label{fig:swit_iv}
\end{figure}

\subsection{Multi-period assignments}\label{sec:mpa}

The same primitives in the {\tt PS} can be used to define potential branches based on many consecutive periods of intervention. We may call these objects $s$-potential branches, and define them in the following alternative to assumption {\tt CP.2}. Naturally, the $0$-potential branch is a potential branch, and so potential branches are a special case of $s$-potential branches. 
\begin{assumption}[{\tt CP.2$^\prime$}] Write the ``$s$-potential branch'' for some $s \geq 0$ at time $t+h$ as $$D_{t,h}(a_{t:t+s}):=\{X_{t,h}(a_{t:t+s}),A_{t,h}(a_{t:t+s}),Y_{t,h}(a_{t:t+s})\},\quad h=0,1,...,H$$ a system counterfactual.  It  corresponds to the assignments at times $t$ through $t+s$ being set to $\{a_{t}, a_{t+1}, \dots, a_{t+s}\}$ and recording the system at horizon $h$ periods later. 
Assume the branch assignment at horizon $h$ is 
$$
A_{t, h}\left(a_{t:t+s}\right):= \begin{cases}
\begin{array}{lll}
     a_{t+h}, & h \leq s \\
     \alpha_{t+h}\left(D_{1:t-1},D_{t,0:h-1}\left(a_{t:t+s}\right), X_{t,h}(a_{t:t+s}),V_{t+h}\right), & h > s.
\end{array}
\end{cases} 
$$
Assume the $s$-potential branch outcome and branch feature at horizon $h$ are
$$
Z_{t,h}(a_{t:t+s}):= \begin{pmatrix}
    X_{t,h}(a_{t:t+s}) \\ Y_{t,h}(a_{t:t+s})
\end{pmatrix} := \begin{cases}
\begin{array}{lll}
     \{X_{t}^\mathtt{T},Y_{t}(A_{1:t-1}, a_{t})^\mathtt{T}\}^\mathtt{T} , & h=0 \\
     \{X_{t+h}(A_{1:t-1}, A_{t,0:h-1}(a_{t:t+s}))^\mathtt{T}, Y_{t+h}(A_{1:t-1}, A_{t,0:h}(a_{t:t+s}))^\mathtt{T}\}^\mathtt{T}, & h > 0.
\end{array}
\end{cases} 
$$
\end{assumption}

Under {\tt CP.2$^\prime$}, analyzing dynamic causal effects in the setting of Example \ref{ex:hlmps}, we see that $D_{t,0}(a_{t:t+s}) = D_{t,0}(a_{t})$ and then recursively, for $h \leq s$, the 
    \[
D_{t,h}(a_{t:t+s})=\left( 
\begin{array}{c}
X_{t,h}(a_{t:t+s}) \\
A_{t,h}(a_{t:t+s}) \\
Y_{t,h}(a_{t:t+s})%
\end{array}%
\right) =\left( 
\begin{array}{c}
\chi_1 \\
0 \\
\gamma_1 + \gamma_{0,X} \chi_1  
\end{array}%
\right) D_{t,h-1}(a_{t:t+s})+\left( 
\begin{array}{c}
0 \\ 
I \\ 
\gamma_{0,A} 
\end{array}%
\right) a_{t+h}+\left( 
\begin{array}{ccc}
\Delta  & 0 & 0\\ 
0 & 0 & 0\\ 
\gamma_{0,X} \Delta & 0& \Omega%
\end{array}%
\right) \varepsilon _{t+h}, 
\]
while for $h>s$,
\[
D_{t,h}(a_{t:t+s}) = \phi D_{t,h-1}(a_{t:t+s}) + B \varepsilon_{t+h}.
\]
Then the dynamic causal effect is again non-stochastic.  For $h \le s$, it is determined by the recursion 
$$
D_{t,h}(a_{t:t+s}) - D_{t,h}(a_{t:t+s}')
= \left( 
\begin{array}{c}
\chi_1 \\
0 \\
\gamma_1 + \gamma_{0,X} \chi_1  
\end{array}%
\right) \{D_{t,h-1}(a_{t:t+s}) - D_{t,h-1}(a_{t:t+s}')\} 
+ \left( 
\begin{array}{c}
0 \\ 
I \\ 
\gamma_{0,A} 
\end{array}%
\right) (a_{t+h} - a_{t+h}')
$$
and then, for $h>s$,   
$$D_{t,h}(a_{t:t+s}) - D_{t,h}(a_{t:t+s}')  = \phi 
\{D_{t,h-1}(a_{t:t+s}) - D_{t,h-1}(a_{t:t+s}')\}.
$$

\subsection{Design-based causal inference}\label{sec:dbci}

Design-based inference is extremely influential in randomized control trials and observational studies (e.g., \cite{Fisher(25),Fisher(35)} and \cite{ImbensRubin(15)}). In these settings, researchers  choose to condition on the potential outcomes.  The importance of the design-based approach in panel data is highlighted in, for example, \cite{ArkhangelskyImbens(24)}.

In time series this design-based approach was introduced by \cite{BojinovShephard(19)} in their simpler setting with no features and randomized assignments. They condition on all the potential outcomes
$$
Y_{1:T}(\mathcal{A}_{1:T}):= \{Y_{1:T}(a_{1:T}):a_{1:T}\in \mathcal{A}_{1:T}\},
$$
where, generically, we write $Z_{s:t}(A_{1:s-1},\mathcal{A}_{s:t}) = \{Z_{s:t}(A_{1:s-1},a_{s:t}):a_{s:t}\in \mathcal{A}_{s:t}\}$.  
One of the attractions of the design-based approach is that it allows some forms of causal inference without specifying a detailed model for the outcomes.  This is very compelling, as time series has no direct form of replication. \cite{LinDing(25)} further develop time series design-based studies, relating it to regression.  

A key object of interest in design-based causal inference is the distribution of assignments conditional on the potential outcomes. This law is naturally called the ``assignment mechanism,'' taken from the cross-sectional literature.   

\begin{definition}[Assignment mechanism]\label{defn:alaw} Assume a {\tt PS}. The assignment mechanism ({\tt AM}) is the law of 
$
A_{1:T} \mid Z_{1:T}(\mathcal{A}_{1:T})
$.
\end{definition}

Our goal is to sample from the {\tt AM} under a particular null hypothesis in order to perform design-based (causal) inference. Towards this goal, we state the following theorem.  

\begin{restatable}{theorem}{thmadecomp}\label{thm:adecomp}
Assume the {\tt PS} from Example \ref{ex:npsem} and that $V_t \ind (U_t, W_t, \varepsilon_{1:t-1}) \mid D_{1:t-1}$ for all $t$. Then the joint law of $
A_{1:T},Z_{1:T}(\mathcal{A}_{1:T}),
$ is determined by the sequence
$$
Z_t(\mathcal{A}_{1:t}) \mid Z_{1:t-1}(\mathcal{A}_{1:t-1})\quad \text{and} \quad 
A_t \mid [A_{1:t-1},Z_{1:t}(A_{1:t-1},\mathcal{A}_t)],
$$
where $t=1,...,T$. 
\end{restatable}
\begin{proof} See the Appendix.\end{proof}

Note that the joint law of $A_{1:T},Z_{1:T}(\mathcal{A}_{1:T})$ is determined by the marginal law of $Z_{1:T}(\mathcal{A}_{1:T})$ and the {\tt AM}. From Theorem \ref{thm:adecomp} and the time series prediction decomposition, the law of $Z_{1:T}(\mathcal{A}_{1:T})$ is entirely determined by the sequence $Z_t(\mathcal{A}_{1:t}) \mid Z_{1:t-1}(\mathcal{A}_{1:t-1})$, and thus the {\tt AM} is solely determined by the sequence $A_t \mid [A_{1:t-1},Z_{1:t}(A_{1:t-1},\mathcal{A}_t)]$. Notice this is close to the {\tt SAM}, but differs as {\tt AM} additionally conditions on $Y_t(A_{1:t-1},\mathcal{A}_t)$.  This observation then motivates the following infeasible theorem.  It uses the assumption that the features are {\tt PS}-exogenous, for all $t$ and $h$, which implies that 
$
X_{1:T}(a_{1:T}) = X_{1:T}$ for all $a_{1:T} \in \mathcal{A}_{1:T}.$

\begin{restatable}{theorem}{thmsampling}
Assume the {\tt PS} from Example \ref{ex:npsem}, that $V_t \ind (U_t, W_t, \varepsilon_{1:t-1})\mid D_{1:t-1}$, that the features are {\tt PS}-exogenous, and that for a sequence of assignments $A^*_{1:T}$ the 
\begin{equation*}
\big[Y_{t}(A^*_{1:t-1},\mathcal{A}_{t}) \ind A_t^*
\big] \mid A^*_{1:t-1},X_{1:t},Y_{1:t-1}(A^*_{1:t-1}).
\end{equation*}  
Then simulating $A_{1:T}^*$ recursively through the conditional law 
\begin{equation}\label{eqn:A*1}
A^*_t \mid [A^*_{1:t-1},X_{1:t},Y_{1:t-1}(A^*_{1:t-1})],\quad t=1,...,T,
\end{equation}
the resulting  
$A_{1:T}^*$ is a draw from the {\tt AM}. 
\end{restatable}
\begin{proof} See the Appendix.\end{proof}

Now the equation (\ref{eqn:A*1}) has the same structure as the sequential assignment mechanism, but for the simulated path of the assignment.  However, the law of (\ref{eqn:A*1}) is still infeasible to sample from as the $Y_{1:t-1}(A^*_{1:t-1})$ are counterfactuals we do not see: we only see the outcomes $Y_{1:t-1} = Y_{1:t-1}(A_{1:t-1})$. 

To make sampling feasible, we define the null composite hypothesis 
$$
H_0: \ 
Y_{t}(a_{1:t}) = Y_{t}(a_{1:t}') + g_t(a_{1:t},a_{1:t}';\theta),\ \forall \ a_{1:t},a_{1:t}' \in \mathcal{A}_{1:t}, \ t=1,...,T,  
$$
where $g_t$ is a deterministic, known, parameterized function where $\theta \in \Theta$.  Under the null, the
$$
Y_{t}(A^*_{1:t}) = Y_{t} + g_t(A^*_{1:t},A_{1:t};\theta),\quad t=1,...,T.
$$
Hence in the setting of Example \ref{ex:npsem}, with the features being {\tt PS}-exogenous, knowledge of the {\tt SAM} allows simulating $A^*_{1:T}$ from the {\tt AM} under the null.

\begin{example}[Homogeneous, linear causal effects]\label{ex:causalPQ} Assume a {\tt PS}, the features are {\tt PS}-exogenous and impose a parametric null hypothesis that 
    $$
    g_t(a_{1:t},a_{1:t}') =  
     \sum_{j=0}^Q\psi_j (a_{t-j} - a_{t-j}')
     + \sum_{j=1}^P \vartheta_j g_{t-j}(a_{1:t-j},a_{1:t-j}'), \quad a_{1:t},a_{1:t}'\in \mathcal{A}_{1:t}, \quad t>Q\ge 0, P\ge 0,
    $$
    so the causal effect is homogeneous and linear in contemporaneous and past assignments with parameters $\theta=(\psi_{0:Q}^{\mathtt{T}},\vartheta_{1:P}^{\mathtt{T}})^{\mathtt{T}}$.  When $\theta=0$ then this is the time series Fisher-type sharp null of no causal effect used by \cite{BojinovShephard(19)}, but now extended to the case of features and causal dynamics.  
 \end{example}
 
    The advantage of this design-based approach is that, for any value of $\theta$, we can simulate under the null $B$ copies of the triple 
    $$
    D_{1:T}^* = \{X_{1:T}^{\mathtt{T}},A^{*\mathtt{T}}_{1:T},Y_{1:T}(A_{1:T}^*)^{\mathtt{T}}\}^{\mathtt{T}},
    $$ 
    which can then be compared to the actual data $D_{1:T}.$  The comparison is made through a low dimensional statistic
    $$
    T(D_{1:T}^*)
    $$
    designed by the researcher, comparing it to  $T(D_{1:T})$: rejecting the null if $T(D_{1:T})$ is large compared to the $B$  simulated versions of $T(D_{1:T}^*)$.  Consequently we can find an exact confidence region $\mathcal{C}$ for $\theta$ (and so in turn causal effects) by inverting the distribution of the test, so that   
    $
    P(\theta \in \mathcal{C}) = 1-\alpha,
    $
    where $\alpha\in (0,1)$ is selected by the researcher.  
    
    A significant virtue of this approach is that we are parametrically modeling the causal effects, but are entirely agnostic about the underlying dynamics of the outcomes.  A downside is that it requires the ability to simulate from {\tt SAM} through (\ref{eqn:A*1}), the law of 
    $$
    A_t^* \mid [A_{1:t-1}^*,X_{1:t},\{Y_{1:t-1}=y_{1:t-1} + g_{1:t-1}(A_{1:t-1}^*,A_{1:t-1})\}].
    $$
    For experimental data this should be known.  
    For observational data, the {\tt SAM} can be estimated from the data, using Assumption {\tt CP.2} to extrapolate from the assignments based on lagged observables (that is Assumption {\tt DGP.2}) to assignments based on lagged counterfactuals. In the important special case of the news impact {\tt PS} this difficulty entirely disappears as $A_t$ is an i.i.d. sequence.

    \begin{example} Suppose there are no confounders, the {\tt PS} is Markovian, the assignments are binary and under the null hypothesis the causal effects follow Example \ref{ex:causalPQ} with $Q=1$ and $P=0$.  Then $a_t \in \{0,1\}$, while 
    $$
    p^*_t(a_t):= P(A_t^*=a_t\mid A_{t-1}^*,Y_{t-1}(A_{t-2:t-1}^*)),
    \quad 
    Y_t(a_{1:t}) = Y_t + \psi_0 (a_t - A_t) + \psi_1 (a_{t-1}-A_{t-1}),
    $$
    so $\theta =\psi_{0:1}$.  
    If $p^*_t=\{p_t^*(a):a \in \{0,1\} \}$ is known, this allows us to calculate $Y_{1:T}(A_{1:T}^*)$ and thus simulate the entire path $$A_{1:T}^*,Y_{1:T}(A_{1:T}^*),$$ with the inputs $\theta$, $Y_{1:T}$ and $A_{1:T}$. Define the inverse probability weighted test statistic 
    \begin{align*}
    {T}(A_{1:t}^*) &:= \frac{1(A_t^*=1) Y_t(A_{1:t-1}^*,1)}{p^*_t(1)}
    - \frac{1(A_t^*=0) Y_t(A_{1:t-1}^*,0)}{p^*_t(0)}.
    \end{align*}
    Then
    \begin{align*}
    \E[{T}(A_{1:t}^*)|Y_{1:T}(\mathcal{A}_{1:T}),A^*_{1:t-1}] &= Y_t(A^*_{1:t-1},1)-Y_t(A^*_{1:t-1},0) = \psi_0.
    \end{align*}
    We compare the observed ${T}(A_{1:t})$ with $B$ independent simulated (under $\theta$) versions.  
    \end{example}

\begin{example}\label{ex:newsdesign}
Assume a news impact \texttt{PS} with 
\begin{align*}
A_{t} &=V_{t}, \\
Y_{t}(a_{1:t}) &=\xi_{t}+f(a_{1:t}),
\end{align*}%
where (i) $V_{t}$ is i.i.d.; (ii) $\xi_{t}$ is strictly stationary in $L^1$;
(iii) $V_{t} \ind \xi_{t}$; and (iv) $f$ is a deterministic function. Then the assumption \texttt{SAM.BR}
holds. This implies that    
\[
\mathrm{E}[Y_{t,h}(a)]=\mathrm{E}[Y_{t+h}\mid A_{t}=a].
\]%
Assume that assignments are binary for all $t$, and define the statistic    
\[
\widehat{T}=T(A_{1:T},Y_{1:T})=W\left( 
\begin{array}{c}
\widehat{\mathrm{E}}[Y_{t}|A_{t}=1]-\widehat{\mathrm{E}}[Y_{t}|A_{t}=0] \\ 
\widehat{\mathrm{E}}[Y_{t+1}|A_{t}=1]-\widehat{\mathrm{E}}[Y_{t+1}|A_{t}=0]
\\ 
\\ 
\widehat{\mathrm{E}}[Y_{t+H}|A_{t}=1]-\widehat{\mathrm{E}}[Y_{t+H}|A_{t}=0]%
\end{array}%
\right) ,
\]%
where $W$ is some conformable, non-stochastic weight matrix and $\hat \E[\cdot \mid \cdot ]$ is an approximation of the conditional expectation. If the assignments are binary, these estimated conditional expectations are just differences in means; otherwise, $\mathrm{E}[Y_{t+h}\mid A_{t}=a]$ may be estimated using, e.g., a kernel
regression. Typically $W$ will be a 
selection matrix, focusing on a single lead.

To carry out inference we condition on all the potential outcomes. We assume the composite null hypothesis and that the causal effects follow Example \ref{ex:causalPQ} with $P=0$ and $Q=H$. Note that under the null $Y_{t,h}(1)-Y_{t,h}(0)=\mathrm{E}[Y_{t,h}(1)]-\mathrm{E}[Y_{t,h}(0)]=\psi_{h}$ almost surely. 

Letting $Y_{1:T}^{\ast }:=Y_{1:T}(A_{1:T}^{\ast })$, under the null hypothesized $\theta$, for $b=1,...,B$, simulate
\[
A_{1:T}^{\ast(b) },Y_{1:T}^{\ast(b) }\quad \text{and} \quad T_{b}^{\ast }:=T(A_{1:T}^{\ast(b)},Y_{1:T}^{\ast(b)}),
\]%
and compute, for example, if $T$ is a scalar, $L_{b}^{\ast }:=\widehat{T}-T_{b}^{\ast }.$  
For a 95\% test, we reject the null based on null hypothesized $\theta$ if 
\[
0\notin \left[ Q_{L^{\ast }}(0.025),Q_{L^{\ast }}(0.975)\right] 
\]%
for $Q_{L^{\ast }}$ the quantile function for $L^*$.
We estimate the quantiles by the sample quantiles from the simulated $%
L_{1}^{\ast },...,L_{B}^{\ast }$. A 95\% CI for $\theta $ are all the values of $\theta \in \Theta $ for
which the associated null is not rejected.  \ 
\end{example}

\subsection{Stochastic dynamic programming}\label{sec:control}

There is a large literature on control, e.g., \cite{AndersonMoore(79)}, \cite%
{Whittle(81),Whittle(82),Whittle(96)},  \cite{Bertsekas(87)} and \cite{HansenSargent(14)}. Here, the discussion will connect the control literature to the potential
system, using the above notation, but with no features. \ Often the
control literature is collected under the label of stochastic dynamic
programming.

Define the sequence which would minimize expected future loss $J_t$ from taking
actions $a_{t:T}\in \mathcal{A}_{t:T}$ as 
\begin{equation*}
\widehat{a}_{t:T|t-1}:={\arg}\underset{a_{t:T} \in \mathcal{A}_{t:T}}{\min} \ \E
[J_{t}(A_{1:t-1},a_{t:T})|D_{1:t-1}],\quad t\in \left\{ 1,2,...,T\right\}
,
\end{equation*}%
give information at time $t-1$. Then take the time-t assignment as 
\begin{equation*}
A_{t}=\widehat{a}_{t|t-1},\quad t=1,...,T,
\end{equation*}%
ignoring the other $\widehat{a}_{t+1:T|t-1}$. Recall that the potential system's ``consistency''
means that 
\begin{equation*}
D_{t}=D_{t}(A_{1:t})=D_{t}(A_{1:t-1},\widehat{a}_{t|t-1}).
\end{equation*}%
Here, the $A_{t}$ is previsible --- it is stochastic but entirely dependent on
past outcome data. \ Thus the {\tt SAM} is
probabilistically degenerate. 

The time $t$ \textquotedblleft
value\textquotedblright\ is defined as 
\begin{eqnarray*}
\mathscr{V}_{t} &:=& \E\left[J_{t}(A_{1:t-1},\widehat{a}_{t:T|t-1})|D_{1:t-1}\right] =\underset{a_{t:T} \in \mathcal{A}_{t:T}}{\min} \ \E
[J_{t}(A_{1:t-1},a_{t:T})|D_{1:t-1}],
\end{eqnarray*}
the best possible expected future loss, given past data. In stochastic dynamic programming, the sequence $\left\{ \widehat{a}_{t|t-1}\right\} $ is
typically called the control sequence, as it is assumed to be under the
control of the researcher. Under the {\tt PS}, it determines the {\tt SAM}. 

The causal effect of moving the assignment to $a_{t}$, away from the optimal
value $a_{t}^{\prime }=A_{t}=\widehat{a}_{t|t-1}$, is particularly
interesting. The immediate dynamic causal effect on the outcome is 
\begin{equation*}
Y_{t,0}(a_{t})-Y_{t,0}(a_{t}^{\prime })=Y_{t}(A_{1:t-1},a_{t})-Y_{t}
\end{equation*}%
which spills over to the effect at horizon $h=1$,
\begin{equation*}
Y_{t,1}(a_{t})-Y_{t,1}(a_{t}^{\prime
})=Y_{t+1}(A_{1:t-1},a_{t},A_{t,1}(a_{t}))-Y_{t+1},
\end{equation*}%
where $A_{t,1}(a_{t})=\widehat{a}_{t+1|t}(a_{t})$ with 
\begin{equation*}
\widehat{a}_{t+1:T|t}(a_{t}):={\arg} \underset{a_{t+1:T} \in \mathcal{A}_{t+1:T}}{\min} \ \E
[J_{t}(A_{1:t-1},a_{t},a_{t+1:T})|D_{1:t-1},D_{t,0}(a_{t})]
\end{equation*}%
being the optimal assignment path from time $t+1$ to time $T$ thinking we had seen
the data $D_{1:t-1},D_{t,0}(a_{t})$. 

In the control context, the researcher typically assumes losses are time
separable, that is 
\begin{equation*}
J_{t}(A_{1:t-1},a_{t:T})=\sum_{j=t}^{T+1}L_{j}(A_{1:t-1},a_{t:j}),\quad
L_{t}(a_{1:t})=\ell_t(Y_{t-1}(a_{1:t-1}),a_{t}),
\end{equation*}%
where the $\ell_t= \{\ell_t(y,a):y \in \mathcal{Y}, a\in \mathcal{A}\}$ are known deterministic functions for all $t=1,...,T$, and in the final period for all $a,a' \in\mathcal{A}$, the $\ell_{T+1}(y,a)=\ell_{T+1}(y,a'):=\ell_{T+1}(y)$. Hence $L_{t}(a_{1:t})$ is random only
because of the potential outcome, while $L_{t+h}(A_{1:t-1},a_{t:t+h})$ is
random because of the assignment sequence and the potential outcome.

Looking forward from time $t$ to time $T$, the overall loss for a sequence
of future assignments $a_{t:T}$ sums up future individual period losses. This can be written as a backward
recursion 
\begin{equation*}
J_{t}(A_{1:t-1}, a_{t:T})=L_{t}(A_{1:t-1}, a_{t})+J_{t+1}(A_{1:t-1}, a_{t},a_{t+1:T}).
\end{equation*}
This recursive structure sets the stage for defining and solving Bellman equations.

\section{Conclusion}\label{sect:conclusion}

This paper defines and works on the potential system ({\tt PS}), a foundational nonparametric model for studying how an assignment at time $t$ causally affects an outcome at time $t+h$, possibly in the presence of confounders. It yields familiar measures of causality --- explored through examples connected to various other literatures in time series causality --- that can be mapped to data-based predictions under familiar assumptions from cross-sectional causal inference. Because this foundational model is defined in terms of low-level nonparametric primitives, it can be readily extended to numerous other time series causality settings, such as design-based based inference, the study of more exotic causal effects, control, and beyond.

This paper does not discuss the intricate details of estimation nor inference --- our focus is on identification. \cite{PlagborgMullerKolesar(24)} and \cite{BallinariWehrli(24)} are population-based inference papers which can sit on top of our {\tt PS}, giving them causal meaning,   covering inference for the relationship between what we call the branch potential outcomes and the assignments.  This builds off inference results in \cite{RambachanShephard(21)}. Other recent work on non-linear impulse response functions include \cite{GoncalvesHerreraKillianPesavento(24),GoncalvesHerreraKillianPesavento(21)} and \cite{GourierouxLee(23)}.

\clearpage

\baselineskip=12pt

\bibliographystyle{chicago}
\bibliography{neil.bib,jake.bib}

\baselineskip=20pt

\clearpage

\section{Appendix}

\subsection{Proofs}

\thmbranchseq*
\begin{proof}
Note that, by the recursion of Example \ref{ex:npsem}, any $Y_{t,h}(a_{t})$ for $h =1,...,H$ only depends on variation in, for any $a_t$,
\[
D_{1:t-1},X_t, W_t, \{\varepsilon_{t+h}\}_{h \geq s \geq 1}
\]
and for $h=0$ depends only on $D_{1:t-1},X_t, W_t,$.

To satisfy condition {\tt SAM.BSU}, it suffices that
$$\big[(D_{1:t-1}, X_t, V_t) \ind (D_{1:t-1},X_t, W_t, \{\varepsilon_{t+s}\}_{h \geq s \geq 1})\big] \mid D_{1:t-1},X_t$$
which reduces to the condition that
$$\big[V_t \ind (W_t, \{\varepsilon_{t+s}\}_{h \geq s \geq 1})\big] \mid D_{1:t-1},X_t.$$
Two conditions that imply this condition, using the contraction property of conditional independence, are
\[
V_t \ind W_t \mid D_{1:t-1}, X_t, \quad V_t \ind \{\varepsilon_{t+s}\}_{h \geq s \geq 1} \mid D_{1:t-1}, X_t, W_t.
\]
However, this second condition is implied by (using the weak union property and recalling that $X_t = \chi_t (D_{1:t-1}, U_t)$)
\[
\{\varepsilon_{t+s}\}_{h \geq s \geq 1} \ind (\varepsilon_t, D_{1:t-1})
\]
which is indeed true by the joint independence of the $\varepsilon_{t}$ across all time. As such, assuming the first property completes the proof.

To satisfy {\tt SAM.BSR} it suffices that
\[
\big[(D_{1:t-1}, X_t, V_t) \ind (D_{1:t-1},X_t, W_t, \{\varepsilon_{t+s}\}_{h \geq s \geq 1})\big] \mid D_{1:t-1}
\]
which reduces to the condition that 
\[
(X_t, V_t) \ind (X_t, W_t,\{\varepsilon_{t+s}\}_{h \geq s \geq 1}) \mid D_{1:t-1}.
\]
By assuming that $X_t$ is $D_{1:t-1}$-measurable, we only require that
\[
V_t \ind (W_t,\{\varepsilon_{t+s}\}_{h \geq s \geq 1}) \mid D_{1:t-1}.
\]
This condition similarly follows from the two other conditions (using the contraction property)
\[
[V_t \ind W_t] \mid D_{1:t-1}, \quad V_t \ind \{\varepsilon_{t+s}\}_{h \geq s \geq 1} \mid W_t, D_{1:t-1}.
\]
The conclusion of this part of the theorem is then immediate using the same arguments as for the first part of the theorem.

If $A_t = V_t$ then to satisfy {\tt SAM.BU} it suffices that
\[
\big[V_t \ind (D_{1:t-1},X_t, W_t, \{\varepsilon_{t+s}\}_{h \geq s \geq 1})\big] \mid X_t
\]
which reduces to the condition that
\[
\big[V_t \ind (D_{1:t-1}, W_t, \{\varepsilon_{t+s}\}_{h \geq s \geq 1})\big] \mid X_t.
\]
By contraction, this condition is implied by the conditions
\[
\big[V_t \ind (D_{1:t-1}, W_t)\big] \mid X_t, \quad \big[V_t \ind \{\varepsilon_{t+s}\}_{h \geq s \geq 1}\big] \mid X_t,D_{1:t-1}, W_t.
\]
The second condition is once again granted using the same argument as for the first part of the theorem proof, and we assume the other.

If $A_t = V_t$ then to satisfy {\tt SAM.BR} it suffices that
\[
\big[V_t \ind (D_{1:t-1},X_t, W_t, \{\varepsilon_{t+s}\}_{h \geq s \geq 1})\big].
\]
By contraction, this condition is implied by the conditions
\[
V_t \ind (D_{1:t-1},X_t, W_t), \quad V_t \ind \{\varepsilon_{t+s}\}_{h \geq s \geq 1} \mid (D_{1:t-1},X_t, W_t).
\]
The second condition is once again granted using the same argument as for the first part of the the theorem proof, and we assume the other.

\end{proof}

\thmbsu*
\begin{proof} We start with the ${\tt CFTE}_{t,h}(a_t,a_t')$ case.  Define the regression function 
\[
\mu_{t,h}^{a_t}(X_t, D_{1:t-1}) := \E[Y_{t+h} \mid A_t = a_t, X_{t},D_{1:t-1}].
\]
Then note that for any $a_t \in \mathcal{A}_t$
\begin{align*}
    \mu_{t,h}^{a_t}(X_t, D_{1:t-1}) &= \E[Y_{t+h} \mid A_t = a_t, X_{t},D_{1:t-1}] \\
    &= \E[Y_{t,h}(a_t) \mid A_t = a_t, X_{t},D_{1:t-1}] \tag{{\tt PS}} \\
    &= \E[Y_{t,h}(a_t) \mid X_{t},D_{1:t-1}]. \tag{{\tt SAM.BSU-}}
\end{align*}
It is then clear that for any $a_t, a_t' \in \mathcal{A}_t$
\[
{\tt CFTE}_{t,h}(a_t,a_t')=\mu_{t,h}^{a_t}(X_t, D_{1:t-1}) - \mu_{t,h}^{a_t'}(X_t, D_{1:t-1}).
\]
The corresponding proofs for ${\tt FTE}_{t,h}(a_t,a_t')$, ${\tt ATE}_{t,h}(a_t,a_t')$, and ${\tt CATE}_{t,h}(a_t,a_t')$ are immediate based on this derivation, following an identical structure.
 
\end{proof}

\thmiv*
\begin{proof}
Recall that, under a one time intervention at time $t$, in the IV {\tt PS}, we can write
\begin{align*}
    A_t = \alpha_t(\tilde{D}_{1:t-1}, X_t,V_t), \quad Y_{t+h}= A_tY_{t,h}(1) + (1-A_t) Y_{t,h}(0) = Y_{t,h}(0) + \{Y_{t,h}(1)-Y_{t,h}(0)\}A_t
\end{align*}
under the {\tt PS} system consistency assumptions. Thus
\begin{align*}
    \E[Y_{t+h} \mid X_t = x_t, \tilde{D}_{1:t-1}] &= \E[Y_{t,h}(0) + \{Y_{t,h}(1)-Y_{t,h}(0)\}A_t \mid X_t = x_t, \tilde{D}_{1:t-1}] \\
    &= \E[Y_{t,h}(0) + \{Y_{t,h}(1)-Y_{t,h}(0)\}A_t(x_t) \mid X_t = x_t, \tilde{D}_{1:t-1}] \\
    &= \E[Y_{t,h}(0) + \{Y_{t,h}(1)-Y_{t,h}(0)\}A_t(x_t) \mid \tilde{D}_{1:t-1}]
\end{align*}
using the fact that $Y_{t,h}(a_t)$ is invariant to the instrument in the IV {\tt PS} and that, using $U_t \ind (V_t, W_t) \mid \tilde D_{1:t-1}$,
\[
X_t \ind \big[A_t(1),A_t(0), Y_{t,h}(1), Y_{t,h}(0)\big] \mid \tilde D_{1:t-1}.
\]
To see this, notice that for $h=0$, this condition reduces to
\[
\chi_t(\tilde D_{1:t-1},U_{t}) \ind \big( \alpha_t(\tilde D_{1:t-1}, 1,V_t), \alpha_t(\tilde D_{1:t-1}, 0,V_t), \gamma_t(\tilde D_{1:t-1},1,W_{t}),\gamma_t(\tilde D_{1:t-1},0,W_{t})\big) \mid \tilde D_{1:t-1}
\]
which is satisfied if $U_t \ind (V_t, W_t) \mid \tilde D_{1:t-1}$;
for $h = 1$, letting $\tilde D_{t}(a_t):=(a_t, \gamma_t(\tilde D_{1:t-1},a_t,W_{t}))$, we have that $$Y_{t,1}(a_t) = \gamma_{t+1}(\tilde D_{1:t-1}, \tilde D_{t}(a_t), \alpha_{t+1}(\tilde D_{1:t-1},\tilde D_{t}(a_t), \chi_{t+1}(\tilde D_{1:t-1},\tilde D_{t}(a_t),U_{t+1}),V_{t+1}) ,W_{t+1}),$$
and by continuing the recursion we see that any $Y_{t,h}(a_t)$ with $h>0$ only depends on variation in $(\tilde D_{1:t-1}, W_t, \{\varepsilon_{t+s}\}_{h\geq s \geq 1})$, so we just (sufficiently) need that
\[
U_{t} \ind \big (V_t, W_t, \{\varepsilon_{t+s}\}_{h\geq s \geq 1}) \mid \tilde D_{1:t-1}
\]
which reduces to the two conditions (by contraction)
\[
U_{t} \ind \big (V_t, W_t) \mid \tilde D_{1:t-1}, \quad U_{t} \ind \big \{\varepsilon_{t+s}\}_{h\geq s \geq 1} \mid \tilde D_{1:t-1}, (V_t, W_t)
\]
where the second condition is satisfied if
\[
\{\varepsilon_{t+s}\}_{h\geq s \geq 1} \ind (\tilde D_{1:t-1}, (V_t, W_t,U_{t}))
\]
which is true by assumption of the SEM {\tt PS}.

Thus
\begin{align*}
    &\E[Y_{t+h} \mid X_t = 1, \tilde{D}_{1:t-1}] - \E[Y_{t+h} \mid X_t = 0, \tilde{D}_{1:t-1}] \\
    = &\E[\{Y_{t,h}(1)-Y_{t,h}(0)\}(A_t(1)-A_t(0)) \mid \tilde{D}_{1:t-1}] \\
    = &\E[\{Y_{t,h}(1)-Y_{t,h}(0)\}\mathbf{1}\{A_t(1)> A_t(0)\} \mid \tilde{D}_{1:t-1}]
\end{align*}
where the last line follows by the monotonicity condition on $\alpha_t$, i.e., that $A_t(1) \geq A_t(0)$ almost surely (in words: in any state of the world $\omega_t \in \Omega_t$, for $\Omega_t$ the underlying sample space of the {\tt PS} random variables at time $t$, the instrument has the same directional effect on the single unit of interest in the time series).

Similarly, we see that
\begin{align*}
    \E[A_{t} \mid X_t = x_t, \tilde{D}_{1:t-1}] &= \E[A_t(x_t) \mid X_t = x_t, \tilde{D}_{1:t-1}]= \E[A_t(x_t) \mid \tilde{D}_{1:t-1}]
\end{align*}
and so
\begin{align*}
    \E[A_{t} \mid X_t = 1, \tilde{D}_{1:t-1}] - \E[A_{t} \mid X_t = 0, \tilde{D}_{1:t-1}] &= \E[(A_t(1)-A_t(0)) \mid \tilde{D}_{1:t-1}] \\
    &=P(A_t(1) > A_t(0) \mid \tilde{D}_{1:t-1})
\end{align*}
again under monotonicity. As such, we can conclude using the law of total expectation that
\[
\E[Y_{t,h}(1) - Y_{t,h}(0) \mid \mathbf{1}\{A_t(1)> A_t(0) \}=1,  \tilde{D}_{1:t-1}] = \frac{E[Y_{t+h} \mid X_t = 1, \tilde{D}_{1:t-1}] - \E[Y_{t+h} \mid X_t = 0, \tilde{D}_{1:t-1}]}{E[A_{t} \mid X_t = 1, \tilde{D}_{1:t-1}] - \E[A_{t} \mid X_t = 0, \tilde{D}_{1:t-1}]}
\]
assuming that $\E[A_{t} \mid X_t = 1, \tilde{D}_{1:t-1}] - \E[A_{t} \mid X_t = 0, \tilde{D}_{1:t-1}] > 0$ almost surely.
\end{proof}

\thmadecomp*
\begin{proof} By the time series prediction decomposition, the assignment mechanism is determined by the sequence of conditional laws of 
$
[A_t,Z_t(\mathcal{A}_{1:T})] \mid [A_{1:t-1},Z_{1:t-1}(\mathcal{A}_{1:T})],
$
for $t=1,...,T$.  Assumptions {\tt CP.1a} and {\tt CP.1b}, non-anticipation and triangularity, implies this simplifies to the conditional law 
$$
[A_t,Z_t(\mathcal{A}_{1:t})] \mid [A_{1:(t-1)},Z_{1:(t-1)}(\mathcal{A}_{1:(t-1)})],\quad t=1,...,T.
$$
In turn this splits into a marginal and conditional law
$$
Z_t(\mathcal{A}_{1:t}) \mid [A_{1:(t-1)},Z_{1:(t-1)}(\mathcal{A}_{1:(t-1)})]\quad \text{and} \quad 
A_t \mid [A_{1:(t-1)},Z_{1:t}(\mathcal{A}_{1:t})].
$$
These simplify to the stated result using the generating structure from Example \ref{ex:npsem} of the {\tt PS} and using that $V_t \ind (U_t, W_t, \varepsilon_{1:t-1}) \mid D_{1:t-1}$. 

To see this, first note that
\[
Z_t(\mathcal{A}_{1:t}) \ind A_{1:(t-1)} \mid Z_{1:(t-1)}(\mathcal{A}_{1:(t-1)})
\]
because the variation in $Z_t(\mathcal{A}_{1:t})$ only depends on $Z_{1:(t-1)}(\mathcal{A}_{1:(t-1)}), U_t, W_t$, and $U_t, W_t$ are independent over time. 

For the second law, notice that, letting $\mathcal{A}_{1:t-1}^- := \mathcal{A}_{1:t-1}\setminus \{a_{1:t-1}^*\}$ for some arbitrary path $a_{1:t-1}^*$ and letting $\{z_{1:t-1}^{a_{1:t-1}},z_t^{a_{1:t}},...\}$ index dummy variables,
\begin{align*}
    &A_t \mid A_{1:t-1}=a_{1:t-1}^*, \{Z_{1:t-1}(a_{1:t-1})= z_{1:t-1}^{a_{1:t-1}}\}_{a_{1:t-1} \in \mathcal{A}_{1:t-1}}, \{Z_t(a_{1:t})=z_t^{a_{1:t}}\}_{a_{1:t} \in \mathcal{A}_{1:t}} \\
    \overset{L}{=}\,\,\,&A_t \mid A_{1:t-1}=a_{1:t-1}^*, \begin{pmatrix}Z_{1:t-1}(a_{1:t-1}^*)= z_{1:t-1}^{a_{1:t-1}^*} \\ \{Z_{1:t-1}(a_{1:t-1})= z_{1:t-1}^{a_{1:t-1}}\}_{a_{1:t-1} \in \mathcal{A}_{1:t-1}^- }\end{pmatrix}, \begin{pmatrix}\{Z_t(a_{1:t-1}^*, a_t)=z_t^{a_{1:t-1}^*,a_t}\}_{a_{t} \in \mathcal{A}_{t}} \\ \{Z_t(a_{1:t-1}, a_t)=z_t^{a_{1:t-1},a_t}\}_{a_{1:t-1} \in \mathcal{A}_{1:t-1}^-,a_{t} \in \mathcal{A}_{t}}\end{pmatrix} \\
    \overset{L}{=}\,\,\,&A_t \mid A_{1:t-1}=a_{1:t-1}^*, \begin{pmatrix}Z_{1:t-1}= z_{1:t-1}^{a_{1:t-1}^*} \\ \{Z_{1:t-1}(a_{1:t-1})= z_{1:t-1}^{a_{1:t-1}}\}_{a_{1:t-1} \in \mathcal{A}_{1:t-1}^- }\end{pmatrix}, \begin{pmatrix}\{Z_t(A_{1:t-1}, a_t)=z_t^{a_{1:t-1}^*,a_t}\}_{a_{t} \in \mathcal{A}_{t}} \\ \{Z_t(a_{1:t-1}, a_t)=z_t^{a_{1:t-1},a_t}\}_{a_{1:t-1} \in \mathcal{A}_{1:t-1}^-,a_{t} \in \mathcal{A}_{t}}\end{pmatrix} \\
    \overset{L}{=}\,\,\,&A_t \mid D_{1:t-1}=d_{1:t-1}, X_t=x_t,\{Z_{1:t-1}(a_{1:t-1})= z_{1:t-1}^{a_{1:t-1}}\}_{a_{1:t-1} \in \mathcal{A}_{1:t-1}^- }, \begin{pmatrix}\{Y_t(A_{1:t-1}, a_t)=y_t^{a_{1:t-1}^*,a_t}\}_{a_{t} \in \mathcal{A}_{t}} \\ \{Z_t(a_{1:t-1}, a_t)=z_t^{a_{1:t-1},a_t}\}_{a_{1:t-1} \in \mathcal{A}_{1:t-1}^-,a_{t} \in \mathcal{A}_{t}}\end{pmatrix}.
\end{align*}
Recalling that $A_t = \alpha_t(D_{1:t-1},X_t,V_t)$, if we can show that
\begin{align*}
    &V_t \mid D_{1:t-1}=d_{1:t-1}, X_t=x_t,\{Z_{1:t-1}(a_{1:t-1})= z_{1:t-1}^{a_{1:t-1}}\}_{a_{1:t-1} \in \mathcal{A}_{1:t-1}^- }, \begin{pmatrix}\{Y_t(A_{1:t-1}, a_t)=y_t^{a_{1:t-1}^*,a_t}\}_{a_{t} \in \mathcal{A}_{t}} \\ \{Z_t(a_{1:t-1}, a_t)=z_t^{a_{1:t-1},a_t}\}_{a_{1:t-1} \in \mathcal{A}_{1:t-1}^-,a_{t} \in \mathcal{A}_{t}}\end{pmatrix} \\
    \overset{L}{=}\,\,\,&V_t \mid D_{1:t-1}=d_{1:t-1}, X_t=x_t, \{Y_t(A_{1:t-1}, a_t)=y_t^{a_{1:t-1}^*,a_t}\}_{a_{t} \in \mathcal{A}_{t}}
\end{align*}
then we have completed the proof. This is true if we have that
\[
V_t \ind \{Z_{1:t}(a_{1:t})\}_{a_{1:t-1} \in \mathcal{A}_{1:t-1}^-,a_{t} \in \mathcal{A}_{t}} \mid  D_{1:t-1}=d_{1:t-1}, X_t=x_t,\{Y_t(A_{1:t-1}, a_t)=y_t^{a_{1:t-1}^*,a_t}\}_{a_{t} \in \mathcal{A}_{t}}
\]
which is granted if
\[
V_t \ind \big(\{Z_{1:t}(a_{1:t})\}_{a_{1:t-1} \in \mathcal{A}_{1:t-1}^-,a_{t} \in \mathcal{A}_{t}}, X_t,\{Y_t(A_{1:t-1}, a_t)\}_{a_{t} \in \mathcal{A}_{t}}\big) \mid D_{1:t-1}
\]
which can be re-written as
\[
V_t \ind \big(\{Z_{1:t}(a_{1:t})\}_{a_{1:t-1} \in \mathcal{A}_{1:t-1}^-,a_{t} \in \mathcal{A}_{t}}, \chi_t(D_{1:t-1}, U_t),\{\gamma_t(D_{1:t-1},\chi_t(D_{1:t-1}, U_t), a_t,W_t)\}_{a_{t} \in \mathcal{A}_{t}}\big) \mid D_{1:t-1}.
\]
Recalling that $V_t$ is independent of the past and that
\[
Z_{t}(a_{1:t}) = \begin{pmatrix}
    \chi_t(D_{1:t-1}(a_{1:t-1}), U_t) \\
    \gamma_t(D_{1:t-1}(a_{1:t-1}),\chi_t(D_{1:t-1}(a_{1:t-1}), U_t), a_t,W_t)
\end{pmatrix}
\]
our proof is then completed by assuming that  $V_t \ind (U_t, W_t, \varepsilon_{1:t-1}) \mid D_{1:t-1}$.

\end{proof}

\thmsampling*
\begin{proof} Use the decomposition result from Theorem \ref{thm:adecomp}, and then we can simulate from 
\[
A^*_t \mid [A^*_{1:t-1},X_{1:t}(A^*_{1:t-1}),Y_{1:t}(A^*_{1:t-1},\mathcal{A}_t)].
\]
The stated result is then immediate from {\tt PS} exogeneity and the additional assumption.

\end{proof}

\end{document}